%% file: paper.tex
\documentclass[a4paper,11pt]{article}
\usepackage{graphicx}
\usepackage{subcaption}
\usepackage[utf8]{inputenc}
\usepackage{amsmath, amsfonts,amsthm}
\usepackage{fullpage}
\usepackage{hyperref}

\usepackage{stmaryrd}

\include{macros}
\title{Practical I/O-Efficient Multiway Separators}

\author{Svend C. Svendsen\thanks{Aarhus University.  Email: \texttt{svendcs@cs.au.dk}}}

\begin{document}

\maketitle

\input{abstract}

\section*{Acknowledgements}
  We thank Pankaj K. Agarwal and Lars Arge for useful discussions on computing planar separators for Delaunay triangulations.
  We also thank Gerth S. Brodal for his input on this paper.

\input{intro}

\input{preliminaries}

\input{algorithm}
\input{terrain}

\input{experiments}

\bibliography{references}
\appendix
\input{appendix-vc-proof}
\input{appendix-sample-boundary}
\input{appendix-separator-stats}

\input{appendix-large-sample}

\end{document}

%% file: macros.tex
\newcommand{\reals}{\mathbb{R}}

\newcommand{\BigO}{O}

\newcommand{\BigTheta}{\Theta}
\newcommand{\BigOmega}{\Omega}
\newcommand{\polylog}{\operatorname{polylog}}

\newcommand{\bigmid}{\mathrel{\big|}}

\newcommand{\Scan}{\operatorname{Scan}}
\newcommand{\Sort}{\operatorname{Sort}}

\newcommand{\vcdim}{\operatorname{dim}_{\text{VC}}}

\newcommand{\varclassifiers}{\mathcal H}
\newcommand{\varclassifier}{\ensuremath{h}}
\newcommand{\varseptree}{\ensuremath{H}}
\newcommand{\varcircles}{\mathcal C}
\newcommand{\vardisks}{\mathcal D}
\newcommand{\vargroundset}{\mathcal X}
\newcommand{\varvcdim}{\xi}
\newcommand{\varsystem}{\Gamma}
\newcommand{\varsamplesystem}{\Upsilon}
\newcommand{\terrain}{\Sigma}
\newcommand{\vargraph}{\mathcal G}
\newcommand{\varregion}{R}
\newcommand{\varraindist}{\mathcal R}
\newcommand{\varvertices}{\mathbb V}

\newcommand{\gb}{\operatorname{GB}}
\newcommand{\tb}{\operatorname{TB}}
\newcommand{\km}{\operatorname{km}}
\newcommand{\kmsquared}{\operatorname{km}^2}

\newcommand\numberthis{\addtocounter{equation}{1}\tag{\theequation}}
\newcommand{\rfig}[1]{Figure~\ref{fig:#1}}
\newcommand{\rsec}[1]{Section~\ref{sec:#1}}
\newcommand{\rapd}[1]{Appendix~\ref{sec:#1}}
\newcommand{\rlem}[1]{Lemma~\ref{lem:#1}}
\newcommand{\rthm}[1]{Theorem~\ref{thm:#1}}
\newcommand{\rtab}[1]{Table~\ref{tab:#1}}
\newcommand{\reqn}[1]{(\ref{eqn:#1})}
\newcommand{\etal}{\textit{et al}.}

\theoremstyle{definition}
\newtheorem{theorem}{Theorem}[section]
\newtheorem{lemma}[theorem]{Lemma}
\newtheorem*{definition}{Definition}

%% file: abstract.tex
\begin{abstract}
We revisit the fundamental problem of I/O-efficiently computing $r$-way separators on planar graphs.
An $r$-way separator divides a planar graph with $N$ vertices into $\BigO(r)$ regions of size
$\BigO(N/r)$ and $\BigO(\sqrt {Nr})$ boundary vertices in total, where boundary vertices are vertices that are adjacent to more than one region.
Such separators are used in I/O-efficient solutions to many fundamental problems on planar graphs such as breadth-first search, finding single-source shortest paths, topological sorting, and finding strongly connected components. 
Our main result is an I/O-efficient sampling-based algorithm that, given a Koebe-embedding of a graph with $N$
vertices and a parameter $r$, computes an $r$-way separator for the graph under certain assumptions on the size of internal memory.
Computing a Koebe-embedding of a planar graph is difficult in practice and no known
I/O-efficient algorithm currently exists.
Therefore, we show how our algorithm can be generalized and applied directly to Delaunay triangulations without relying on a Koebe-embedding.
This adaptation can produce many boundary vertices in the worst-case, however, to our knowledge our result is the first to be implemented in practice due to the many non-trivial and complex techniques used in previous results. 
Furthermore, we show that our algorithm performs well on real-world data and that the number of boundary vertices is small in practice.

Motivated by applications in geometric information systems, we show how our
algorithm for Delaunay triangulations can be applied to
compute the flow accumulation over a terrain, which models how much water flows over the vertices of a terrain.
When given an $r$-way separator, our implementation of the algorithm outperforms traditional sweep-line-based algorithms on the publicly available digital elevation model of Denmark. 
%
\end{abstract}

%% file: intro.tex
\section{Introduction}
\label{sec:intro}
In this paper, we revisit the fundamental problem of computing $r$-way separators by presenting I/O-efficient algorithms and demonstrating how our results can be applied to I/O-efficiently compute flow accumulation on a terrain.
We implement and evaluate our algorithms on real-world terrain data.

The $r$-way separator is a generalization of the \emph{planar separator theorem} by Lipton \etal~\cite{lipton1979separator}.
The planar separator theorem states that a planar graph with $N$ vertices can be partitioned into two unconnected sets each of size at most $(2/3)N$ by removing $\BigO \big(\sqrt N \big)$ vertices from the graph.
Lipton and Tarjan~\cite{lipton1979separator} showed that such a partitioning can be computed in linear time in classical models of computation.
Frederickson \etal~\cite{frederickson1987separator} described how the planar separator theorem can be generalized to the concept of an \emph{$r$-way separator}:
Given a parameter $r$, an $r$-way separator is a division of the vertices of the graph into $\BigO(r)$ non-disjoint \emph{regions} such that each vertex of the graph is contained in at least one region.
A region contains two types of vertices: boundary vertices and interior vertices.
An \emph{interior vertex} is contained in exactly one region and is adjacent only to vertices in that region.
A \emph{boundary vertex} is shared among at least two regions and is adjacent to vertices in multiple regions.
Each region contains $\BigO(N/r)$ vertices in total of which $\BigO(\sqrt{N/r})$ are boundary vertices.
It follows that the total number of boundary vertices is $\BigO(\sqrt{Nr})$.

The concept of $r$-way separators is particularly interesting when handling planar graphs that exceed the capacity of the main memory since the computation of such separators can be used to divide the graph into memory-sized regions.
In this situation, data is written and read in large blocks to disk, so it is important to design algorithms that minimize the movement of such blocks.
This has led to the development of the so-called \emph{I/O model} by Aggarwal and Vitter~\cite{aggarwal88complexity}.
In this model, the computer is equipped with a two-level memory hierarchy consisting of an \emph{internal memory} capable of holding $M$ data items and an \emph{external memory} of unlimited size.
All computation has to happen on data in internal memory.
Data is transferred between internal and external memory in blocks of $B$ consecutive data items.
Such a transfer is referred to as an \emph{I/O-operation} or an \emph{I/O}.
The cost of an algorithm is the number of I/Os it performs.
The number of I/Os required to read $N$ consecutive items from disk is
$\Scan(N) = \BigO(N/B)$ and the number of I/Os required to sort $N$ items is
$\Sort(N) = \BigTheta \big((N/B) \log_{M/B}(N/B)\big)$~\cite{aggarwal88complexity}.
For all realistic values of $N$, $M$ and $B$ we have $\Scan(N) < \Sort(N) \ll N$.
Maheshwari \etal~\cite{maheshwari2008separator} showed that an $r$-way separator can be computed in $\BigO(\Sort(N))$ I/Os.
This algorithm results in solutions to fundamental graph problems, such as breadth-first search, finding single-source shortest paths, topological sorting, and finding strongly connected components, that uses $\BigO(\Sort(N))$ I/Os~\cite{arge2003topo,arge2003components}.
Later, Arge \etal~\cite{arge2013separator} presented an I/O-efficient $r$-way separator algorithm that uses $\BigO(\Sort(N))$ I/Os and $\BigO(N \log N)$ internal memory computation time.
Furthermore, they showed that this result can be used to derive algorithms for finding single-source shortest paths, topological sorting, and finding strongly connected components using $\BigO(\Sort(N))$ I/Os and $\BigO(N \log N)$ internal memory computation time.
This improves upon the result by Maheshwari \etal~by upper bounding the internal memory computation time used.

To our knowledge, no algorithms for I/O-efficiently computing multiway separators have been implemented in practice yet due to the many non-trivial and complex techniques used to derive them.
Therefore, we consider the problem of computing $r$-way planar separators when given a \emph{Koebe-embedding} of the graph.
A \emph{Koebe-embedding} of a planar graph is a set of disks in the plane with disjoint interiors where the center of each disk corresponds to a vertex in the graph and two disks are tangent if and only if the corresponding vertices in the graph are adjacent.
Miller \etal~\cite{miller1997separator} showed that a Koebe-embedding can be used to partition the corresponding graph into two unconnected parts each of size at most $(3/4)N$ by removing $\BigO\big(\sqrt N \big)$ vertices from the graph.
In this paper, we present a simple I/O-efficient algorithm that computes an $r$-way separator for a  planar graph when given a \emph{Koebe-embedding} of the graph and having certain assumptions on the size of internal memory.

To our knowledge, the computation of Koebe-embeddings is not trivial and no I/O-efficient algorithms have been presented.
Bannister \etal~\cite{bannister2014complexity} showed that computing exact Koebe-embedding requires computing the roots of polynomials of unbounded degree.
Thus, the focus of the current state-of-the-art algorithms is to numerically approximate the Koebe-embedding.
Orick \etal~\cite{orick2017circlepacking} presented an algorithm that approximates a Koebe-embedding by alternating between adjusting radii and positions of vertices.
Empirical results show that the algorithm runs in approximately linear time, however, no theoretical worst-case bounds are given.
Recently, Dong \etal~\cite{dong2020circlepacking} presented an algorithm based on convex optimization that computes an approximate Koebe-embedding in near-linear worst-case time.
To our knowledge, these algorithms do not trivially extend to the I/O model.

Motivated by applications in geometric information systems, we show that our algorithm can be adapted to Delaunay triangulation without having to first compute a Koebe-embedding.
Delaunay triangulations can be computed using $\BigO(\Sort(N))$ I/Os~\cite{agarwal2005delaunay} and are widely used to convert terrain point clouds into so-called triangulated irregular networks which represent a terrain as a triangulated surface.
A \emph{triangulated irregular network} (\emph{TIN}) is computed by projecting the terrain point cloud in $\reals^3$ onto the $xy$-plane, computing the Delaunay triangulation of the projected points, and lifting the Delaunay triangulation back to $\reals^3$.
This adaptation can result in $\BigOmega(N)$ boundary vertices in the worst case~\cite{miller1995delaunay}.
However, we test our algorithm on the publicly available digital elevation model of Denmark~\cite{sdfe2021dem} and show that the algorithm results in a small number of boundary vertices in practice.

Finally, we describe how $r$-way separators can be used to compute the \emph{flow accumulation} over a terrain, which models the flow of water over a terrain represented as a TIN.
We consider a variant of the flow accumulation problem, where we are given a rain distribution function $\varraindist$ that fits in internal memory and assigns $\varraindist(v) \geq 0$ units of water to each vertex $v$.
The water in each vertex $v$ is then distributed by pushing water to a neighboring vertex according to a given flow direction of $v$.
The flow accumulation of a vertex $v$ is the total amount of water that flows through $v$.
This problem is traditionally solved in the I/O-model using $\BigO(\Sort(N))$ I/Os by a sweep-line algorithm where the flow is propagated using a priority queue during a downward sweep of the terrain~\cite{haverkort2012grid}.
In this paper, we adapt the grid terrain algorithm by Haverkort \etal~\cite{haverkort2012grid} to speed up the computation of flow accumulation over a TIN when given an $r$-way separator of the terrain.
Furthermore, we show that this algorithm performs well in practice and outperforms the traditional sweep-line algorithm on the digital elevation model of Denmark when given an $r$-way separator of the terrain.

%% file: preliminaries.tex
\section{Preliminaries}
\label{sec:preliminaries}
In this section, we state several preliminary definitions and introduce a
more general definition of the $r$-way separator.

\subsection{k-ply Neighborhood Systems} 
We begin by presenting a more formal definition of a Koebe-embedding and then
introduce the more general $k$-ply neighborhood system.
This generalization will be used later when applying our result to Delaunay
triangulations.
The definitions in this section follow Miller \etal~\cite{miller1997separator}.
Let a \emph{disk packing} be a set of disks $\{B_1, \ldots, B_N\}$ in the plane that have disjoint interiors.
Koebe~\cite{koebe1936embedding} showed that every planar graph can be embedded as a disk packing such that the center of each disk corresponds to a vertex in the planar graph and two disks are tangent if and only if there is an edge connecting the two corresponding vertices in the graph.
We refer to this as a \emph{Koebe-embedding} of the planar graph.
Note that a partitioning of a Koebe-embedding into disjoint subsets implies a partitioning of the vertices of the graph.
Miller \etal~\cite{miller1997separator} used this idea to describe how a Koebe-embedding of a planar graph can be used to compute a planar separator.
In order to describe this result, we first introduce the more general $k$-ply neighborhood system:
\begin{definition}[$k$-ply neighborhood system]
  A $k$-ply neighborhood system in $d$ dimensions is a set $\varsystem = \{ B_1, \ldots, B_n
  \}$ of closed balls in $\reals^d$ such that no point in $\reals^d$ is
  in the interior of more than $k$ of the balls.
\end{definition}

In the following sections, we introduce the notion of a separator for $k$-ply neighborhood systems in $\reals^d$ for general $k$.
Observe that a disk packing is a $1$-ply neighborhood system and, thus, this separator will also be applicable to Koebe-embeddings.

We now state the planar separator result by Miller \etal~\cite{miller1997separator}.
A $d$-dimensional sphere $\varclassifier$ partitions a $k$-neighborhood system $\varsystem$ in $\reals^d$ into three subsets:
the set $\varsystem(\varclassifier_{>})$ of all balls of $\varsystem$ contained in the exterior of $\varclassifier$, the set $\varsystem(\varclassifier_{<})$ of all balls of $\varsystem$ contained in the interior of $\varclassifier$, and the set $\varsystem(\varclassifier_{=})$ of all balls of $\varsystem$ that intersect
the boundary of $\varclassifier$.
Correspondingly, we define the subsets $\varsystem(\varclassifier_{\leq}) = \varsystem(\varclassifier_{<}) \cup \varsystem(\varclassifier_{=})$ and $\varsystem(\varclassifier_{\geq}) = \varsystem(\varclassifier_{>}) \cup \varsystem(\varclassifier_{=})$.
\begin{theorem}[Sphere Separator~\cite{miller1997separator}]
  Suppose $\varsystem$ is a $k$-ply neighborhood system in $\reals^d$ with size $|\varsystem|$. Then there exists a sphere $\varclassifier$ in $\reals^d$ such that
  \begin{align*}
    |\varsystem(\varclassifier_<)|, |\varsystem(\varclassifier_>)| &\leq \frac{d+1}{d+2} \cdot |\varsystem|\;,\\
    |\varsystem(\varclassifier_=)| &= \BigO \big(k^{1/d} \cdot |\varsystem|^{1 - 1/d}\big)\;.
  \end{align*}
\end{theorem}
Additionally, Miller \etal~\cite{miller1997separator} presented a sampling-based algorithm for approximately computing sphere separators.
We state their result with two additional properties that follow from their original proof; first, we state the result when applied to a subset $\varsamplesystem \subseteq \varsystem$. Note that $\varsamplesystem$ does not have to be a proper subset.
Secondly, we state the number of I/Os used.
\begin{theorem}[Randomized Separator Algorithm~\cite{miller1997separator}]
  Suppose $\varsystem$ is a $k$-ply neighborhood system in $\reals^d$, $\varsamplesystem \subseteq \varsystem$ is a subset of $\varsystem$.
  Then for any constant $\varepsilon>0$ we can compute a sphere~$\varclassifier$ such that with probability at least $1/2$
  \begin{align*}
    |\varsamplesystem(\varclassifier_<)|, |\varsamplesystem(\varclassifier_>)| &\leq \left( \frac{d+1}{d+2} + \varepsilon \right) |\varsamplesystem|\;,\\
    |\varsystem(\varclassifier_=)| &= \BigO\big(k^{1/d} \cdot |\varsystem|^{1-1/d}\big)\;.
  \end{align*}
  Furthermore, the algorithm uses $\BigO(\Scan(|\varsamplesystem| \cdot d) + c_2)$ I/Os, where $c_2$ is a constant depending only on $\varepsilon$ and $d$.
  \label{thm:miller-separator}
\end{theorem}
We remark that the randomization in the algorithm is over random numbers chosen by the algorithm independent of the input.
Therefore, the algorithm can be used to find a sphere satisfying the inequalities by applying the algorithm an expected constant number of times~\cite{miller1997separator}.
When describing our algorithm, we will use \rthm{miller-separator} as a black box and refer to the resulting sphere as a \emph{sphere separator}.

\subsection{Multiway Separator}
We now present a generalization of the sphere separator result by Miller \etal~\cite{miller1997separator}.
Given a $k$-ply neighborhood system $\varsystem$ in $\reals^d$ and a parameter $r \leq |\varsystem|/k$,
an \emph{$r$-way division} of $\varsystem$ is a division of $\varsystem$ into $\BigO(r)$ non-disjoint \emph{regions} such that each ball in $\varsystem$ is contained in at least one region.
A region contains two types of balls: boundary balls and interior balls.
An \emph{interior ball} is contained in exactly one region and has non-empty intersection only with balls contained in the same region.
A \emph{boundary ball} is shared among at least two regions and has non-empty intersection with balls in multiple regions.
Each region contains $\BigO(|\varsystem|/r)$ balls in total which are stored consecutively on disk.
An \emph{$r$-way separator} is an $r$-way division where each region contains $\BigO \big(\sqrt{k|\varsystem|/r} \big)$ boundary balls.
It follows that the total number of boundary balls of an $r$-way separator is $\BigO \big( \sqrt{k|\varsystem|r} \big)$.
We use the term \emph{multiway separator} and \emph{multiway division} whenever $r$ is clear from the context.

\subsection{Range Spaces, VC dimensions, and Samples}
The main result of this paper is obtained by computing a multiway separator on a sample of a given $k$-ply neighborhood system.
In order to prove correctness of our algorithm, we show that the result generalizes to the entire neighborhood system with at least constant probability.
This proof relies on the concepts of Vapnik–Chervonenkis dimension (VC dimension)~\cite{har2011geometric} and relative $\varepsilon$-approximations~\cite{har2011relative}.
Here, we will provide a quick summary of various definitions and theorems.
For a more in-depth introduction to VC-dimension, we refer to Har-Peled \etal~\cite{har2011geometric}.
\begin{definition}[Range Space]
  A \emph{range space} is a pair $(\vargroundset, \varclassifiers)$, where
  $\vargroundset$ is the ground set (finite or infinite) and $\varclassifiers$
  is a (finite or infinite) family of subsets of $\vargroundset$. The elements
  of $\varclassifiers$ are referred to as \emph{classifiers}.
\end{definition}
\begin{definition}[VC Dimension]
  Let $S = (\vargroundset, \varclassifiers)$ be a range space.
  Given $Y \subseteq \vargroundset$, let the \emph{intersection} of $Y$ and $\varclassifiers$ be defined as
  $$Y \cap \varclassifiers = \left\{ \varclassifier \cap Y \mid \varclassifier \in \varclassifiers \right\}\;.$$
  If $Y \cap \varclassifiers$ contains all subsets of $Y$, then we say that $Y$ is \emph{shattered} by $\varclassifiers$.
  The \emph{VC Dimension} of $S$, denoted by $\vcdim(S)$, is the maximum cardinality of a shattered subset of $\vargroundset$:
  $$ \vcdim(S) = \max \left\{ |Y| \bigmid Y \subseteq \vargroundset \wedge |Y \cap \varclassifiers| = 2^{|Y|} \right\}\;. $$
  If there are arbitrarily large shattered subsets, then $\vcdim(S) = \infty$.
\end{definition}
\begin{lemma}[VC Dimension of Halfspaces~{\cite[Chapter 5]{har2011geometric}}]
  Let $S = (\vargroundset, \varclassifiers_\text{halfspace})$ be the range space where $\vargroundset = \reals^d$ and $\varclassifiers_\text{halfspace}$ is the set of halfspaces in $\reals^d$.
  Then $S$ has VC~dimension~$d+1$.
  \label{lem:vc-halfspace}
\end{lemma}
\begin{lemma}[Mixing of Range Spaces~{\cite[Chapter 5]{har2011geometric}}]
  Let $S_1 = (\vargroundset, \varclassifiers_1), \ldots, S_k = (\vargroundset, \varclassifiers_k)$ be $k$ range spaces which share the same ground set $\vargroundset$ and all have VC dimension at most $\varvcdim$.
  Consider the sets of classifiers $\varclassifiers_\cap$ and $\varclassifiers_\cup$, where
  \begin{align*}
    \varclassifiers_\cap &= \{ \varclassifier_1 \cap \cdots \cap \varclassifier_k \mid \varclassifier_1 \in \varclassifiers_1, \ldots, \varclassifier_k \in \varclassifiers_k \}\;, \\
    \varclassifiers_\cup &= \{ \varclassifier_1 \cup \cdots \cup \varclassifier_k \mid \varclassifier_1 \in \varclassifiers_1, \ldots, \varclassifier_k \in \varclassifiers_k \}\;.
  \end{align*}
  Then the range spaces $S_\cap = (\vargroundset, \varclassifiers_\cap)$ and $S_\cup = (\vargroundset, \varclassifiers_\cup)$ have VC dimension $\BigO(\varvcdim k \log k)$.
  \label{lem:vc-compose}
\end{lemma}
\begin{definition}[Measure]
  Let $S = (\vargroundset, \varclassifiers)$ be a range space, and let $X \subseteq \vargroundset$ be a finite subset of~$\vargroundset$.
  The \emph{measure} of a classifier $\varclassifier \in \varclassifiers$ in $X$ is the quantity
  \begin{align*}
    \bar X(\varclassifier) = \frac{|\varclassifier \cap X|}{|X|}\;.
  \end{align*}
\end{definition}
\begin{definition}[Relative Approximation]
  Let $S = (\vargroundset, \varclassifiers)$ be a range space, and let $X \subseteq \vargroundset$ be a finite subset of $\vargroundset$.
  For given parameters $0<p,\varepsilon<1$, a subset $Y \subseteq X$ is a \emph{relative $(p,\varepsilon)$-approximation} for $(X,S)$ if, for each $\varclassifier \in \varclassifiers$, we have
  \begin{align*}
    (1-\varepsilon) \bar X(\varclassifier) \leq \bar Y(\varclassifier) \leq (1+\varepsilon) \bar X(\varclassifier) &\quad\text{if } \bar X(\varclassifier) \geq p\;. \\
    \bar X(\varclassifier) - \varepsilon p \leq \bar Y(h) \leq \bar X(\varclassifier) + \varepsilon p &\quad\text{if } \bar X(\varclassifier) < p\;.
  \end{align*}
\end{definition}
\begin{lemma}[Relative Approximation Sampling~\cite{har2011relative}]

  Let $S = (\vargroundset, \varclassifiers)$ be a range space with VC dimension $\varvcdim$, and let $X \subseteq \vargroundset$ be a finite subset of $\vargroundset$.
  Given parameters $0< p,\varepsilon,q <1$, a random sample $Y \subseteq X$ of size at least
  \begin{align*}
    \frac{c}{\varepsilon^2 p} \left(\varvcdim \log \frac 1p + \log \frac 1q \right)\;,
  \end{align*}
  for an appropriate constant $c$, is a relative $(p,\varepsilon)$-approximation
  for $(X,S)$ with probability at least $1-q$.
  \label{lem:vc-sample-size}
\end{lemma}
%

%% file: algorithm.tex
\section{Multiway Separator Algorithm for \textit{k}-ply Neighborhood Systems}
\label{sec:algorithm}
In this section, we state our main result for I/O-efficiently computing an $r$-way separator of a $k$-ply neighborhood system $\varsystem$.
The algorithm can be applied to $k$-ply neighborhood systems in $\reals^d$ for any dimensions $d>2$, however, we prove correctness only for $d=2$.

We begin by presenting an algorithm that computes an $r$-way division of $\varsystem$ under the assumption that $k \leq \log \frac MB \log \log \frac MB$.
Given a $k$-ply neighborhood system $\varsystem$ in the plane and a parameter $r$, we let $\hat r = \min(r, \lfloor M/B \rfloor)$ and compute an $r$-way division by recursively computing $\hat r$-way divisions until $\varsystem$ is divided into regions of size $\BigO(|\varsystem|/r)$.
In order to compute an $\hat r$-way division on $\varsystem$, we sample a subset $\varsamplesystem \subseteq \varsystem$ of sufficiently large size.
By recursively computing sphere separators using \rthm{miller-separator}, we can compute an $\hat r$-way separator for $\varsamplesystem$.
Let $\varseptree$ denote the sphere separators that are computed during the recursion.
We refer to $\varseptree$ as a \emph{separator tree}.
We prove that with at least constant probability we obtain an $r$-way division by recursively applying the sphere separators of $\varseptree$ on $\varsystem$.
It follows that we obtain an $\hat r$-way division for $\varsystem$ by repeating this sampling-based algorithm an expected constant number of times.
This result provides guarantees on the number of boundary balls in the sample $\varsamplesystem$, however, we do not prove bounds for the total number of boundary balls in the $r$-way division of $\varsystem$.
We expect the number of boundary balls to be small and confirm so by experimental evaluation in later sections.
Additionally, by increasing the sample size and slightly modifying the algorithm, one can remove the assumption on $k$ and prove that the result is an $r$-way separator for $\varsystem$.
This results in an I/O-efficient algorithm for computing $r$-way separators when $\log^3 \frac MB \log \log \frac MB \log \frac{|\varsystem|}{k} = \BigO\big(\sqrt{Mk}\big)$.
In other words, we provide guarantees on the number of boundary balls by assuming $M$ is sufficiently large.

The rest of this section is structured as follows:
in \rsec{algorithm:sample}, given $\hat r \leq M/B$, we describe how to sample $\varsamplesystem$, recursively apply \rthm{miller-separator}, and prove that the result can be used to divide $\varsystem$ into regions of size $\BigO(|\varsystem|/\hat r)$ with at least constant probability.
In \rsec{algorithm:total-boundary}, we bound the number of regions to $\BigO(\hat r)$ and the total number of boundary balls in $\varsamplesystem$ to $\BigO(\sqrt{k|\varsystem|\hat r})$.
In \rsec{algorithm:region-boundary}, we bound the number of boundary balls in each region of $\varsamplesystem$ to $\BigO(\sqrt{k |\varsystem|/ \hat r})$.
Finally, in \rsec{algorithm:complexity}, we bound the expected number of I/Os used and state the final algorithm.

\subsection{Recursively Computing Separators}
\label{sec:algorithm:sample}
In this subsection, we describe how to sample $\varsamplesystem$, recursively apply \rthm{miller-separator}, and prove that the result can be used to divide $\varsystem$ into regions of size $\BigO(|\varsystem|/\hat r)$ with at least constant probability, where $\hat r \leq M/B$ and assuming $k \leq |\varsamplesystem|/\hat r$.

First, sample $\varsamplesystem \subseteq \varsystem$ of size at least $c_0 \cdot \hat r \log^2 \hat r \log \log \hat r$, where $c_0>0$ is a  constant we choose later.
Letting $l = \BigO(\log \hat r)$, we recursively compute sphere separators on $\varsamplesystem$ for at most $l$ levels;
let $\varsamplesystem_i \subseteq \varsamplesystem$ denote the balls of $\varsamplesystem$ that occur in a node $i$ of the recursion.
In the root of the recursion, we let $\varsamplesystem_i = \varsamplesystem$.
At each node of the recursion, we compute a sphere separator $\varclassifier$ such that $\varsamplesystem_i(\varclassifier_\leq)$ and $\varsamplesystem_i(\varclassifier_\geq)$ are smaller than $\varsamplesystem_i$ by at least a constant factor.
For now, assume that such a sphere separator $h$ is obtained.
We then recurse on the two subproblems $\varsamplesystem_i(\varclassifier_\leq)$ and $\varsamplesystem_i(\varclassifier_\geq)$.
The recursion is continued until the problem size is at most $c \cdot (|\varsamplesystem|/\hat r)$, where $c>0$ is a sufficiently large constant.
The separator tree $\varseptree$ is then formed from the sphere separators by letting the nodes in $\varseptree$ correspond to the recursively computed sphere separators.

We proceed by describing how to compute a sphere separator $h$ in a node $i$ of the recursion.
Using \rthm{miller-separator} and setting $\varepsilon = \frac {1}{12}$, we compute a sphere separator $\varclassifier$ that with probability at least $1/2$ satisfies
\begin{align}
  \label{eqn:sample-inside-size} |\varsamplesystem_i(\varclassifier_<)|, |\varsamplesystem_i(\varclassifier_>)| &\leq \frac {10}{12} \cdot |\varsamplesystem_i|\;, \\
  \label{eqn:sample-intersect-size} |\varsamplesystem_i(\varclassifier_=)| &\leq c_1 \sqrt{k|\varsamplesystem_i|}\;,
\end{align}
where $c_1 > 0$ is a constant.
Note that this uses $\BigO(\Scan(|\varsamplesystem_i|))$ I/Os.
We apply \rthm{miller-separator} an expected constant number of times until a separator $\varclassifier$ that satisfies \reqn{sample-intersect-size} and \reqn{sample-inside-size} is obtained.
Since we divide $\varsamplesystem$ into regions of size at most $c \cdot |\varsamplesystem|/\hat r$, it follows that $|\varsamplesystem|/\hat r \leq |\varsamplesystem_i|/c$.
Using the assumption that $k \leq |\varsamplesystem|/\hat r$ and that \reqn{sample-intersect-size} holds, we upper bound $|\varsamplesystem_i(\varclassifier_=)|$ as follows:
\begin{align}
  |\varsamplesystem_i(\varclassifier_=)|
    \leq c_1 \sqrt{k|\varsamplesystem_i|}
    \leq c_1 \sqrt{\frac{|\varsamplesystem|}{\hat r}|\varsamplesystem_i|}
    \leq c_1 \sqrt{\frac{|\varsamplesystem_i|^2}{c}}
    \leq \frac{c_1}{\sqrt c} |\varsamplesystem_i| \;.
  \label{eqn:sample-intersect-bound}
\end{align}
Thus, for $\sqrt c \geq 12 \cdot c_1$, it follows from \reqn{sample-intersect-bound} that $|\varsamplesystem_i(\varclassifier_=)| \leq \frac{1}{12} |\varsamplesystem_i|$.
Combining this with \reqn{sample-inside-size}, we obtain a separator $\varclassifier$ that satisfies
\begin{align}
  |\varsamplesystem_i(\varclassifier_\geq)|,|\varsamplesystem_i(\varclassifier_\leq)| \leq \frac{11}{12} |\varsamplesystem_i| \;.
  \label{eqn:sample-inside-bound}
\end{align}
Thus, the problem size becomes smaller by a constant factor and we can recursively compute separators $\varclassifier$ until $\varsamplesystem$ is divided into regions of size at most $c \cdot |\varsamplesystem|/\hat r$.
This requires at most $l = \BigO(\log r)$ levels of recursion.

We proceed by showing that $\varsystem$ is divided into regions of size $\BigO(|\varsystem|/\hat r)$ when $\varsystem$ is divided recursively using the sphere separators of $\varseptree$.
We proceed introducing the following two lemmas:
\begin{lemma}
  Let $\varcircles$ be the set of all circles in the plane and let $\vardisks$ be the set of all disks in the plane.
  Let $\varclassifiers_{\leq}$ be the set of classifiers defined as
  $\varclassifiers_{\leq} = \{ \vardisks(\varclassifier_\leq) \mid \varclassifier \in \varcircles \}$.
  Correspondingly, we define $\varclassifiers_{\geq}$.
  The range spaces $(\vardisks, \varclassifiers_\leq)$ and $(\vardisks, \varclassifiers_\geq)$ have constant VC dimension.
  \label{lem:vc-circle}
\end{lemma}
The proof of \rlem{vc-circle} is included in \rapd{appendix-vc-proof}.
\begin{lemma}
  Let $\varclassifiers_l$ be the set of classifiers defined as
  \begin{align*}
    \varclassifiers_l &= \big\{ \varclassifier_1 \cap \cdots \cap \varclassifier_l \mid \varclassifier_1, \ldots, \varclassifier_l \in (\varclassifiers_\leq \cup \varclassifiers_\geq) \big\}\;.
  \end{align*}
  The range space $(\vardisks, \varclassifiers_l)$ has VC dimension $\BigO(l \log l)$.
  \label{lem:vc-compose-spheres}
\end{lemma}
\begin{proof}
  Observe that any finite subset $Y \subset \vardisks$ shattered in the range space $(\vardisks, \varclassifiers_\leq \cup \varclassifiers_\geq)$ can also be shattered in the range space $(\vardisks, \varclassifiers_\cup)$, where $\varclassifiers_\cup = \{\varclassifier_1 \cup \varclassifier_2 \mid \varclassifier_1 \in \varclassifiers_\leq, \varclassifier_2 \in \varclassifiers_\geq \}$.
  Thus the VC-dimension of $(\vardisks, \varclassifiers_\leq \cup \varclassifiers_\leq)$ is upper bounded by the VC-dimension of $(\vardisks, \varclassifiers_\cup)$.
  The proof now follows from \rlem{vc-circle} and \rlem{vc-compose}.
\end{proof}

Observe that the separator tree $\varseptree$ defines a set of regions such that each region is defined by the intersection of at most $l$ classifiers in the set $(\varclassifiers_\leq \cup \varclassifiers_\geq)$ corresponding to the sphere separators in a path from the root to a leaf in $\varseptree$.
Thus, a region can be defined by a classifier in $\varclassifiers_l$.
Furthermore, each region contains at most $c \cdot (|\varsamplesystem|/\hat r)$ balls of $\varsamplesystem$.
It now follows from \rlem{vc-compose-spheres} and \rlem{vc-sample-size}, that by sampling $\varsamplesystem$ with size at least $c_0 \cdot \hat r \log^2 \hat r \log \log \hat r$, $\varsamplesystem$ is a relative $(1/\hat r, \varepsilon)$-approximation of $\varsystem$ in the range space $(\vardisks, \varclassifiers_l)$ with at least constant probability.
The constant $c_0>0$ is chosen according to \rlem{vc-compose-spheres} and \rlem{vc-sample-size}.
Thus, with at least constant probability, the regions of $\varsystem$ contains at most $\BigO
\big(\frac{|\varsystem|}{|\varsamplesystem|} c \cdot (|\varsamplesystem|/\hat r)\big) = \BigO(|\varsystem|/\hat r)$ balls.

\subsection{Bounding the Total Number of Boundary Balls}
\label{sec:algorithm:total-boundary}
In the previous subsection, we bounded the size of each region.
However, the number of regions may be large, since boundary balls occur in multiple regions.
Recall that in a node $i$ of the recursion we obtain a sphere separator $\varclassifier$ such that the number of intersected balls is $|\varsamplesystem_i(\varclassifier_=)| \leq c_1 \sqrt{k |\varsamplesystem_i|}$.
We proceed to upper bound the total number of boundary balls by bounding the total number of intersected balls during the recursion.
We show how to bound the total number of intersections in $\varsamplesystem$ by $\BigO \big(\sqrt{k|\varsamplesystem|\hat r}\big)$.

In \rsec{algorithm:sample}, we argued that the recursion on $\varsamplesystem$ produces regions of size at most $a = c \frac{|\varsamplesystem|}{\hat r}$, where $c>0$ is a constant.
Furthermore, it follows from \reqn{sample-inside-bound} that the size of the smallest region is at least $(1/12) \cdot a$.
Similar to Frederickson~\cite{frederickson1987separator}, we let $b(|\varsamplesystem|)$ denote the number of intersections of balls in $\varsamplesystem$ during the recursion.
At each node $i$ of the recursion, $\varsamplesystem_i$ is divided into two regions $\varsamplesystem_i(\varclassifier_\leq)$ and $\varsamplesystem_i(\varclassifier_\geq)$ by the sphere separator $\varclassifier$.
Recall that $\varclassifier$ is selected such that $\varsamplesystem_i(\varclassifier_\leq) \leq \beta |\varsamplesystem_i|$, where $(1/12) \leq \beta \leq (11/12)$.
Furthermore, it follows from \reqn{sample-intersect-size} that $\varsamplesystem_i(\varclassifier_\geq) \leq (1 - \beta) \cdot |\varsamplesystem_i| + c_1 \sqrt{k|\varsamplesystem_i|}$.
We upper bound $b(|\varsamplesystem_i|)$ as follows:
\begin{equation*}
    b(|\varsamplesystem_i|) \leq \begin{cases}
      b \left(\beta |\varsamplesystem_i|\right) + b \left((1-\beta) \cdot |\varsamplesystem_i| + c_1 \sqrt{k|\varsamplesystem_i|}\right) + c_1 \sqrt{k|\varsamplesystem_i|} \quad &\text{if}\;|\varsamplesystem_i| > a  \\
      0 \quad &\text{if}\; \frac {1}{12} \cdot a \leq |\varsamplesystem_i| \leq a \\
     \end{cases}
\end{equation*}
It can be shown by induction in the size of $\varsamplesystem_i$ that $b(|\varsamplesystem_i|) = \BigO \big(\sqrt{k|\varsamplesystem_i|\hat r} \big)$.
The proof of this is included in \rapd{appendix-sample-boundary}.
Using the assumption $k \leq |\varsamplesystem|/\hat r$, it follows that the number of regions in $\varsamplesystem$ and $\varsystem$ is $\BigO \Big(\frac{|\varsamplesystem|+b(|\varsamplesystem|)}{(1/12) \cdot a}\Big) = \BigO\Big(\frac{|\varsamplesystem|}{a}\Big) = \BigO(\hat r)$.
Thus, the separator tree $\varseptree$ can be used to divide $\varsystem$ into $\BigO(\hat r)$ regions of size $\BigO(|\varsystem|/\hat r)$ by recursively applying the sphere separators of $\varseptree$ to $\varsystem$

\subsection{Reducing the Number of Boundary Balls in a Region}
\label{sec:algorithm:region-boundary}
In order to obtain an $r$-way separator of $\varsamplesystem$ from an $r$-way division of $\varsamplesystem$, we reduce the number of boundary balls in each region.
Let $b_i$ be the number of boundary balls in a region $\varregion_i$ of $\varsamplesystem$.
Letting $c_2>0$ be a constant, we describe how to ensure $b_i \leq c_2 \cdot \sqrt{k |\varsamplesystem|/\hat r}$ for each region $\varregion_i$.
We do this by adapting the technique described by Arge \etal~\cite[Section 3.3]{arge2013separator}.
It follows from the previous subsection that $\sum b_i = \BigO \big(\sqrt{k |\varsamplesystem| \hat r}\big)$.
Let $g$ denote the total number of regions which contain more than $c_2 \sqrt{k |\varsamplesystem|/\hat r}$ boundary balls.
For each region $\varregion_i$ where $b_i > c_2 \sqrt{k |\varsamplesystem|/\hat r}$, we recursively apply \rthm{miller-separator} on the boundary balls of $\varregion_i$.
In other words, we recursively compute sphere separators $\varclassifier$ that divide $\varregion_i$ into regions $\varregion_i(\varclassifier_\leq)$ and $\varregion_i(\varclassifier_\geq)$ with at most $(10/12) \cdot b_i + c_3 \sqrt{k |\varregion_i|}$ boundary balls each, where $c_3>0$ is a constant.
Recall that $|\varregion_i| \leq c \frac{|\varsamplesystem|}{\hat r}$, where $c>0$ is a constant.
It follows that for sufficiently large $c_2$, the region is divided into two regions such that the number of boundary balls in each is $(11/12) \cdot b_i$.
We recurse until the number of boundary balls is at most $c_2 \sqrt{k |\varsamplesystem|/\hat r}$.
Observe, that the total number of sphere separators required to divide all $g$ regions is
\begin{align*}
  \BigO \left(\sum \frac{b_i}{c_2 \sqrt{k |\varsamplesystem|/\hat r}}\right) = \BigO\left(\frac{\sqrt{k |\varsamplesystem| \hat r}}{c_2 \sqrt{k |\varsamplesystem|/\hat r}}\right) = \BigO(\hat r) \;.
\end{align*}
Thus, the total number of regions will be increased by only $\BigO(\hat r)$.
Furthermore, since each sphere separator adds $\BigO(\sqrt{k|\varsamplesystem|/\hat r})$ new boundary balls, the total number of boundary balls remains $\BigO(\sqrt{k|\varsamplesystem|\hat r})$.
Each separator can be computed using expected $\BigO(\Scan(|\varsamplesystem|/\hat r))$ I/Os and, thus, we can reduce the number of boundary balls in each region by using an additional $\BigO(\Scan(|\varsamplesystem|))$ I/Os.
Thus, the algorithm results in an $\hat r$-way separator for $\varsamplesystem$.

\subsection{Bounding the Total I/O-Complexity}
\label{sec:algorithm:complexity}
We now bound the total I/O-complexity of the algorithm described.
Let $\varseptree$ be the separator tree computed during the $\BigO(\log \hat r)$ levels of the recursion.
Recall, at each node $i$ of the recursion we compute a sphere separator $\varclassifier$ using \rthm{miller-separator} using expected $\Scan(|\varsamplesystem_i|)$ I/Os.
Thus, the expected number of I/Os used in a level of recursion is $\BigO \big(\Scan(|\varsamplesystem|+b(|\varsamplesystem|))\big)$.
Using the upper bound on $b(|\varsystem|)$ and the assumption $k \leq |\varsamplesystem|/\hat r$, we conclude that the total number of I/Os used during the $\BigO(\log \hat r)$ levels of recursion is $\BigO \big(\Scan(|\varsamplesystem|) \log \hat r\big)$.
The results can now be stated in the following lemma:
\begin{lemma}
  Given a $k$-ply neighborhood system $\varsystem$ in the plane, a parameter $\hat r \leq M/B$, and a random sample $\varsamplesystem$ such that $|\varsamplesystem| \geq c_0 \cdot \hat r \log^2 \hat r \log \log \hat r$ and $k \leq |\varsamplesystem|/\hat r$, where $c_0$ is a constant.
  Then an $\hat r$-way separator $\varseptree$ for $\varsamplesystem$ can be computed using expected $\BigO \big(\Scan(|\varsamplesystem|) \log \hat r \big)$ I/Os.
  Furthermore, with at least constant probability $\varseptree$ produces an $\hat r$-way division of $\varsystem$ when applied to $\varsystem$.
  \label{lem:sample-computation}
\end{lemma}

In order to compute the $\hat r$-way division of $\varsystem$, we apply \rlem{sample-computation} to obtain a separator tree $\varseptree$ such that the sphere separators of $\varseptree$ can be used to divide $\varsystem$ into $\BigO(\hat r)$ regions of size $\BigO(|\varsystem|/\hat r)$ with at least constant probability.
Note that since $\hat r = \lfloor M/B \rfloor$, the number of leaves in $\varseptree$ is $\BigO(M/B)$.
We partition $\varseptree$ into a constant number of subtrees such that each subtree $\hat \varseptree$ fits in memory along with one block per leaf of $\hat \varseptree$.
We now divide $\varsystem$ by recursing over the subtrees using $\BigO(\Scan(|\varsystem|))$ I/Os in total.
We repeat the above an expected constant number of times until an $\hat r$-way division of $\varsystem$ is obtained.

We proceed by bounding the number of I/Os used by \rlem{sample-computation} to $\BigO(\Scan(\varsystem))$.
We use a sample of size $c_0 \frac MB \log^2 \frac MB \log \log \frac MB$ and assume $|\varsystem| > \frac MB \log^3 \frac MB \log \log \frac MB$.
Under this assumption, the expected number of I/Os is bounded as follows:
\begin{align*}
  \Scan(|\varsamplesystem|) \log \hat r &= \BigO \left(\frac{\hat r \log^3 \hat r \log \log \hat r}{B} \right)
                          = \BigO \left(\frac {M}{B^2} \log^3 \frac MB \log \log \frac MB \right)\\
                          &= \BigO \left(\frac{|\varsystem|}{B} \right) = \BigO \big(\Scan(|\varsystem|) \big)\;.
\end{align*}
Next, consider the case when $|\varsystem| \leq \frac MB \log^3 \frac MB \log \log \frac MB$.
We observe that it is sufficient to divide $\varsystem$ into regions of size $\BigO(M)$, since regions of size $\BigO(M)$ can be divided further by directly applying \rthm{miller-separator} on the regions.
That is, each region of size $\BigO(M)$ fits in memory after a constant number of applications of \rthm{miller-separator}, so directly applying \rthm{miller-separator} uses $\BigO(\Scan(|\varsystem|))$ additional I/Os.
Thus, using \rlem{sample-computation}, we compute a separator tree $\varseptree$ such that $\varseptree$ can be used to divide $\varsystem$ into $\BigO(\bar r)$ regions of size $\BigO(M)$ with at least constant probability, where
\begin{align*}
  \bar r=\frac 1B \log^3 \frac MB \log \log \frac MB \;.
\end{align*}
It follows that a sample $\bar \varsamplesystem$ of size $\BigO(\polylog (M/B)) = \BigO(M)$ is sufficient.
Similar to before, we repeat the sampling and separator computation until an $\bar r$-way division of $\varsystem$ is found.
This uses expected $\BigO(\Scan(|\varsystem|))$ I/Os since the problem size fits in internal memory after a constant number of levels of recursion.

\begin{lemma}
  Given a $k$-ply neighborhood system $\varsystem$ in the plane and a parameter $\hat r \leq M/B$ such that $k \leq \log^2 \frac MB \log \log \frac MB$, an $\hat r$-way division of $\varsystem$ can be computed using expected $\BigO \big(\Scan(|\varsystem|)\big)$ I/Os.
  \label{lem:base-case}
\end{lemma}

We now present the final algorithm for computing an $r$-way division for $\varsystem$ and general $r$.
This algorithm follows immediately from \rlem{base-case}.
Let $\hat r = \min(r, \lfloor M/B \rfloor)$ and recursively apply \rlem{base-case} for $\BigO(\log_{M/B}(r))$ levels of recursion.
This uses a total of $\BigO \big(\Scan(|\varsystem|) \log_{M/B}(r) \big)$ I/Os.
\begin{theorem}
  Given a $k$-ply neighborhood system $\varsystem$ in the plane, an $r$-way division of $\varsystem$ can be computed using expected $\BigO \big(\Scan(|\varsystem|) \log_{M/B}(r) \big)$ I/Os, assuming $k \leq \log^2 \frac MB \log \log \frac MB$.
  \label{thm:recursive-case}
\end{theorem}
Furthermore, this implies an algorithm for Koebe-embeddings, since Koebe-embeddings are $1$-ply neighborhood systems.
\begin{theorem}
  Given a planar graph $\vargraph$ and a Koebe-embedding $\varsystem$ of $\vargraph$,
  an $r$-way division of $\vargraph$ can be computed using expected $\BigO \big(\Scan(|\varsystem|) \log_{M/B}(r) \big)$ I/Os.
  \label{thm:koebe-algorithm}
\end{theorem}

Note that the above results in an $r$-way division, but does not provide any bounds on the number of boundary balls when dividing $\varsystem$.
In \rapd{appendix-large-sample}, we argue that the number of boundary balls in $\varsystem$ can be bounded by increasing the sample size.
The result is stated in the theorem below.

\begin{theorem}
  Given a $k$-ply neighborhood system $\varsystem$ in the plane,
  an $r$-way separator of $\varsystem$ can be computed using expected $\BigO \big(\Scan(|\varsystem|) \log_{M/B}(r) \big)$ I/Os, assuming $k \leq |\varsystem|/r$ and
  \begin{align*}
    \log^3 \frac MB \log \log \frac MB \log \frac{|\varsystem|}{k} = \BigO \Big(\sqrt{Mk}\Big) \;.
  \end{align*}
  \label{thm:large-sample}
\end{theorem}

%% file: terrain.tex
\section{Applications to Delaunay Triangulations and Terrain}
\label{sec:terrain}
Motivated by applications in geographic information systems, we turn our attention to Delaunay triangulations.
Delaunay triangulations are often used in practice to compute TINs from point clouds and can be computed I/O-efficiently using $\BigO(\Sort(N))$ I/Os~\cite{agarwal2005delaunay}.
In the previous section, we described how an $r$-way separator for a planar graph $\vargraph$ can be computed I/O-efficiently, when given a Koebe-embedding of $\vargraph$ and having certain assumptions on the size of the internal memory.
However, to our knowledge, there are no known I/O-efficient algorithms for the computation of Koebe-embeddings.
Therefore, we describe how to adapt our algorithm to compute separators for triangulations without computing a Koebe-embedding.
This adaptation can also be applied to any triangulation in the plane and, thus, any planar graph since a planar graph $\vargraph$ can be triangulated by trivially adding edges until every face of $\vargraph$ is a triangle.

Observe that the circumcircles of a triangulation $\vargraph$ in the plane form a $k$-ply neighborhood system $\varsystem$ in the plane.
However, $k$ is $\BigOmega(|\vargraph|)$ in the worst-case since many circumcircles may overlap in one point.
Miller \etal~\cite{miller1995delaunay} showed that the circumcircles of a Delaunay triangulation form a $k$-ply neighborhood with at most constant $k$ when the largest ratio of the circum-radius to the length of the smallest edge over all triangles is at most constant.
Therefore, we describe how to adapt our algorithm to circumcircles of a Delaunay triangulation.
We compute an $r$-way division $\varseptree$ for $\varsystem$ and use $\varseptree$ to divide $\vargraph$ by mapping each region $\varregion_i$ of $\varseptree$ to a region $\hat \varregion_i$ as follows:
Let the boundary vertices of $\hat \varregion_i$ be the vertices that are contained in a triangle whose circumcircle is a boundary ball of $\varregion_i$.
The internal vertices of $\hat \varregion_i$ are the vertices which are contained only in triangles whose circumcircles are internal balls of region $\varregion_i$.

We see that the number of vertices in $\hat \varregion_i$ is $\BigO(|\varregion_i|)$, where $|\varregion_i|$ is the number of circumcircles of $\varregion_i$.
Furthermore, the number of boundary vertices in $\hat \varregion_i$ regions is at most a constant factor larger than the number of boundary circumcircles of $\varregion_i$.

It follows that we can also divide a TIN into regions by applying the above construction to the Delaunay triangulation used to construct the TIN.
This approach allows for division-based algorithms on TINs.
Haverkort \etal~\cite{haverkort2012grid} showed that a separator of a grid graph can be used to I/O-efficiently compute the flow accumulation of a terrain.
Their results can be applied to TINs and $r$-way divisions:
\begin{lemma}[Division-Based Flow Accumulation~\cite{haverkort2012grid}]
  Let $\terrain$ be a TIN and let $\varvertices$ denote the vertices of the TIN.
  Given a rain distribution $\varraindist : \varvertices \rightarrow \reals^+$ and an $r$-way division of $\terrain$ such that each region fits in internal memory along with the rain distribution.
  Letting $b_i$ denote the number of boundary vertices in a region $\varregion_i$ of the $r$-way division,
  the flow accumulation over $\terrain$ can be computed using $\BigO(\Scan(|\varvertices|) + \Sort(\sum b_i))$ I/Os
\label{lem:division-flow-accumulation}
\end{lemma}
The flow accumulation over a TIN can also be computed using an algorithm that performs a top-down sweep of the terrain~\cite{haverkort2012grid}.
The result is stated in the following lemma:
\begin{lemma}[Sweep-Based Flow Accumulation~\cite{haverkort2012grid}]
  Let $\terrain$ be a TIN and let $\varvertices$ denote the vertices of the TIN.
  Given a rain distribution $\varraindist : \varvertices \rightarrow \reals^+$ that fits in internal memory, the flow accumulation of $\terrain$ can be computed using $\BigO(\Sort(|\varvertices|))$ I/Os.
  \label{lem:sweep-flow-accumulation}
\end{lemma}

%% file: experiments.tex
\section{Experiments}
\label{sec:experiments}
In this section, we present the experiments we have conducted on real terrain data to demonstrate the efficiency of our algorithms.
We have implemented and tested the $r$-way division algorithm for Delaunay triangulations as described in \rsec{terrain} and \rthm{recursive-case}.
Furthermore, we have implemented and tested the two flow accumulation algorithms stated in \rlem{division-flow-accumulation} and \rlem{sweep-flow-accumulation}.
The algorithms have been implemented in C++ and make heavy use of the TPIE library~\cite{arge2017tpie} which provides implementations of fundamental I/O-efficient algorithms such as sorting and priority queues.
The experiments were run on an Intel i7-3770 CPU with $32 \gb$ of RAM and $28 \tb$ of storage running Linux.
Each program was assigned $25 \gb$ of memory during testing.

For our tests, we used the Danish Elevation Model published by the Danish Agency for Data Supply and Efficiency~\cite{sdfe2021dem}.
The model consists of a highly detailed point cloud collected using LiDAR technology.
Each point in the model is labeled as ground, vegetation, a building rooftop, and others.
There is an average of 4.5 points per square meter, however, this varies for each area of the terrain since non-reflective surfaces, such as certain types of vegetation, can result in large holes in the point cloud.
In this paper, we focus on modeling the flow of water and, thus, we filtered out points not labeled as ground or building points.
We constructed a TIN from the resulting point cloud using a Delaunay triangulation.

We have implemented \rthm{miller-separator} to compute sphere separators as described by Miller \etal~\cite{miller1997separator} and Clarkson \etal~\cite{clarkson1993radon}.
Whenever we sample sphere separators on input $\vargraph$, we discard a separator $\varclassifier$ if it does not satisfy that $|\vargraph(\varclassifier_\leq)| \geq 1/4$ and $|\vargraph(\varclassifier_\geq)| \geq 1/4$, where $|\vargraph|$ is the number of triangles in $\vargraph$, $|\vargraph(\varclassifier_\leq)|$ is the number of triangles inside or intersecting $\varclassifier$, and $|\vargraph(\varclassifier_\geq)|$ is the number of triangles outside or intersecting $\varclassifier$.
We note that this differs from previous sections, where we computed the separator based on the circumcircles.
However, since the number of intersected circumcircles upper bounds the number of intersected triangles, the previous proofs still apply to this adaptation.
Furthermore, this simplifies the algorithm by avoiding the computation of circumcircles.
In order to examine the size of $|\vargraph(\varclassifier_=)|$, we sampled a large number of sphere separators on various areas of the terrain and computed the number of intersected triangles.
The results are shown in \rapd{appendix-separator-stats} and demonstrate that $|\vargraph(\varclassifier_=)|$ is small in practice.

Next, we assessed our $r$-way division algorithm (\rthm{recursive-case}) and the flow accumulation algorithms described in \rlem{division-flow-accumulation} and \rlem{sweep-flow-accumulation}.
We evaluated our implementation on a TIN representing a $3500 \kmsquared$ area around Herning, Denmark.
The TIN is represented as a list of triangles, such that each vertex of a triangle is annotated with its coordinates, an index, and the flow direction of the vertex.
The flow directions have been computed by selecting the neighboring vertex with minimum height.
This resulted in a TIN with $~23.3$ billion points and a total size of $6.2 \tb$.
In order to measure the I/O throughput of the system, we implemented a simple program that reads the entire TIN and observed that the program took $2$ hours and $56$ minutes to run on the TIN.

We computed a multiway division on the input such that our implementation of division-based flow accumulation (\rlem{division-flow-accumulation}) can fit each region in memory. 
The division was saved to disk by representing each region as a file containing a list of triangles.
Each triangle is annotated with the same information as the input as well as a boolean variable indicating whether the triangle is a boundary triangle or not.
This resulted in a division with $131$ regions and only $9\,394\,816$ boundary vertices in total.
Furthermore, $\sum b_i = 17\,257\,717$, where $b_i$ is the number of boundary vertices in region $\varregion_i$.
Thus, we see that the number of boundary vertices is very small in practice despite having no theoretical bounds for this algorithm.
The computation of this multiway division took $19$ hours and $23$ minutes.
We implemented a program that reads the output division and writes a dummy division with the same number of regions and the same number of triangles in each region.
This program took $3$ hours and $16$ minutes to read the division and $8$ hours and $12$ minutes to write the dummy division.
In other words, we see that a large proportion of the running time is due to computation in internal memory.

Having computed the division, we applied it to compute the flow accumulation over the TIN using the division-based algorithm (\rlem{division-flow-accumulation}).
In our implementation, we used the uniform rain distribution that distributes one unit of water on each vertex.
In this setup, our implementation of division-based flow accumulation took $13$ hours and $59$ minutes when given the division as input.
We compared this to an implementation of \rlem{sweep-flow-accumulation} which took $20$ hours and $18$ minutes when running on the same input TIN.
Thus, we see a significant improvement in running time when the terrain has been preprocessed by computing the multiway division.
In other words, our approach improves performance when computing flow accumulation for different rain distributions on the same TIN.
However, the division-based approach is slower when computing the flow accumulation for only a single rain distribution, and the time spent preprocessing is included.

%% file: appendix-vc-proof.tex
\section{Proof of Constant VC dimension}
\label{sec:appendix-vc-proof}
In this section, we present the proof of \rlem{vc-circle}.
\begin{proof}
  We begin by proving that the VC dimension of $(\vardisks, \varclassifiers_\geq)$ is constant.
  Let $\varclassifier_\geq \in \varclassifiers_\geq$ denote a classifier corresponding to the separator sphere $\varclassifier$ with center $(x_\varclassifier, y_\varclassifier)$ and radius $r_\varclassifier$.
  Let $d \in \vardisks$ be a disk in the plane with center $(x_d, y_d)$ and radius $r_d$.
  We can express $d \in \vardisks(\varclassifier_\geq)$ as follows:
  \begin{align*}
    d \in \vardisks(h_\geq)
      &\Leftrightarrow \sqrt{(x_\varclassifier - x_d)^2 + (y_\varclassifier - y_d)^2} \geq r_\varclassifier - r_d \\
      &\Leftrightarrow r_d \geq r_\varclassifier \vee (x_\varclassifier - x_d)^2 + (y_\varclassifier - y_d)^2 \geq (r_\varclassifier - r_d)^2 \;.
  \end{align*}
  We now map $d$ to $\reals^7$ by mapping the monomials $x_d^iy_d^jr_d^k$ to variables $z_{i}$ as follows:
  \begin{align*}
    \big \langle z_{1}, z_{2},z_{3}, z_{4}, z_{5}, z_{6}, z_{7} \big \rangle = \big \langle x_d, y_d, x_dy_d, x_d^2, y_d^2, r_d, r_d^2 \big \rangle\;.
  \end{align*}
  Assume there is a finite subset $Y \subset \vardisks$ that is shattered by $\varclassifiers_\geq$.
  It follows that $Y$ maps to a subset $\hat Y \subset \reals^7$ where $|Y| = |\hat Y|$ and $\hat Y$ is shattered by $\varclassifiers_\cup = \{ h_1 \cup h_2 \mid h_1, h_2 \in \varclassifiers_\textit{halfspace} \}$, where $\varclassifiers_\textit{halfspace}$ is the set of all halfspaces in $\reals^7$.
  Thus, the VC dimension of $(\vardisks, \varclassifiers_\geq)$ is at most the VC dimension of $(\reals^7, \varclassifiers_\cup)$.
  It follows from \rlem{vc-halfspace} and \rlem{vc-compose} that the VC dimension of $(\vardisks, \varclassifiers_\geq)$ is constant.
  The VC dimension of $(\vardisks, \varclassifiers_\leq)$ can be bounded correspondingly.
\end{proof}

%% file: appendix-sample-boundary.tex
\section{Bound on the Total Number of Boundary Balls}
\label{sec:appendix-sample-boundary}
In this section, we provide an upper bound for the function $b(|\varsamplesystem_i|)$ introduced in \rsec{algorithm:total-boundary}. Recall the definition of $b(|\varsamplesystem_i|)$:
\begin{equation*}
    b(|\varsamplesystem_i|) \leq \begin{cases}
      b \left(\beta |\varsamplesystem_i|\right) + b \left((1-\beta) \cdot |\varsamplesystem_i| + c_1 \sqrt{k|\varsamplesystem_i|}\right) + c_1 \sqrt{k|\varsamplesystem_i|} \quad &\text{if}\;|\varsamplesystem_i| > a  \\
      0 \quad &\text{if}\; \frac {1}{12} \cdot a \leq |\varsamplesystem_i| \leq a \\
     \end{cases}
\end{equation*}
We proceed to show that $b(|\varsamplesystem_i|) \leq \frac{c_4 |\varsamplesystem_i| \sqrt {k}}{\sqrt a} - c_5 \sqrt{k|\varsamplesystem_i|}$, where $c_5 = 5 \cdot (c_1 + 1/100)$ and $c_4 = c_5 \cdot \sqrt{12}$.
We prove this by induction in the size of $\varsamplesystem_i$.
Letting $|\varsamplesystem_i| \geq \frac{1}{12} \cdot a$, it follows that $\frac{c_4 |\varsamplesystem_i| \sqrt {k}}{\sqrt a} - c_5 \sqrt{k|\varsamplesystem_i|} \geq 0$ since $c_5 \leq c_4\sqrt{1/12}$.
This proves the base case.
We proceed with proving the induction step using the induction hypothesis:
\begin{align*}
b(|\varsamplesystem_i|)
  &\leq b \Big (\beta |\varsamplesystem_i| \Big) + b \Big((1-\beta) \cdot |\varsamplesystem_i| + c_1\sqrt{k|\varsamplesystem_i|}\Big) + c_1 \sqrt{k |\varsamplesystem_i|} \\
  &\leq \frac{c_4 \left(|\varsamplesystem_i| + c_1\sqrt{k|\varsamplesystem_i|}\right) \sqrt {k}}{\sqrt a} - c_5 \sqrt{|\varsamplesystem_i| k \beta} - c_5 \sqrt{|\varsamplesystem_i| k (1 - \beta)} + c_1 \sqrt{k |\varsamplesystem_i|} \\
  &\leq \frac{c_4|\varsamplesystem_i| \sqrt {k}}{\sqrt a} - \sqrt {k|\varsamplesystem_i|} \left (c_5 \sqrt \beta + c_5 \sqrt{1 - \beta} - c_1 - \frac{c_4c_1 \sqrt {k}}{\sqrt a} \right) \\
  &\leq \frac{c_4|\varsamplesystem_i| \sqrt {k}}{\sqrt a} - c_5 \sqrt {k|\varsamplesystem_i|} \left (\sqrt \beta + \sqrt{1 - \beta} - \frac{c_1}{c_5} - \frac{c_4c_1 \sqrt{k}}{c_5 \sqrt a} \right)\;.
\end{align*}
Recall the assumption that $k \leq |\varsamplesystem|/\hat r$ and that $a= c \frac{|\varsamplesystem|}{\hat r}$. It follows that $k/a \leq 1/c$.
Furthermore, since $(1/12) \leq \beta \leq (11/12)$ it follows that $\sqrt \beta + \sqrt {1-\beta} \geq \frac 65$.
We now rewrite as follows:
\begin{align*}
  b(|\varsamplesystem_i|)
    &\leq \frac{c_4|\varsamplesystem_i| \sqrt {k}}{\sqrt a} - c_5 \sqrt {k |\varsamplesystem_i|} \left (\frac 65 - \frac{c_1}{c_5} - \frac{c_4c_1 \sqrt{k/a}}{c_5} \right) \\
    &\leq \frac{c_4|\varsamplesystem_i| \sqrt{k}}{\sqrt a} - c_5 \sqrt {k |\varsamplesystem_i|} \left (\frac 65 - \frac{c_1 + c_4c_1 \sqrt {1/c}}{c_5} \right) \;.
\end{align*}
Recall that the recursion produces regions of size at most $a = c \frac{|\varsamplesystem|}{\hat r}$, where $c>0$ is a constant.
By inserting $c_5 = 5 \cdot (c_1 + 1/100)$ and setting $c$ sufficiently large such that $c_4c_1\sqrt{1/c} \leq 1/100$, we obtain
\begin{align*}
  b(|\varsamplesystem_i|)
    &\leq \frac{c_4|\varsamplesystem_i| \sqrt{k}}{\sqrt a} - c_5 \sqrt {k |\varsamplesystem_i|} \left (\frac 65 - \frac{c_1 + \frac{1}{100}}{5 \cdot (c_1 + \frac{1}{100})} \right) \\
    &= \frac{c_4|\varsamplesystem_i| \sqrt{k}}{\sqrt a} - c_5 \sqrt {k |\varsamplesystem_i|}\;.
\end{align*}
This completes the proof by induction.
Thus, the total number of intersections when recursing $\varsamplesystem$ is at most $b(|\varsamplesystem_i|) \leq \frac{c_4|\varsamplesystem_i| \sqrt {k}}{\sqrt a} = \BigO \big(\sqrt{k|\varsamplesystem_i|\hat r} \big)$.

%% file: appendix-separator-stats.tex
\section{Experimental Evaluation of Separator Size}
\label{sec:appendix-separator-stats}
In this section, we examine the number of intersections when computing sphere separators on a terrain represented as a TIN.
As previously described in \rsec{experiments}, we have implemented \rthm{miller-separator} as described by Miller
\etal~\cite{miller1997separator} and Clarkson \etal~\cite{clarkson1993radon}.
We apply the implementation to various areas of a TIN constructed from the Danish Elevation model including only ground and building points.
Let $\vargraph$ denote the input TIN and let $|\vargraph|$ denote the number of triangles in $\vargraph$.
For each area we sample $250$ sphere separators such that each sphere separator $\varclassifier$ satisfies $|\vargraph(\varclassifier_\leq)| \geq 1/4$ and $|\vargraph(\varclassifier_\geq)| \geq 1/4 \cdot |\vargraph|$, where $|\vargraph|$ is the number of triangles in $\vargraph$, $|\vargraph(\varclassifier_\leq)|$ is the number of triangles inside or intersecting $\varclassifier$, and $|\vargraph(\varclassifier_\geq)|$ is the number of triangles outside or intersecting $\varclassifier$.


\input{figures/intersection-types}
\input{tables/intersection-types-points-distribution}

We first examine how the separator size changes when the type of terrain changes.
For this, we examine three different $10 \km$ by $10 \km$ tiles of the terrain.
The first tile is an area from the city of Copenhagen which has a high frequency of building points.
The second tile is an area centered on the town of Gammel Rye which has a high frequency of points labeled as high vegetation.
Finally, the third tile is centered on the town of Fjeldstervang which has a high frequency of ground points describing agricultural fields.
We have described the distribution of point labels in \rtab{intersection-types-point-distribution}.
For each tile, we have plotted a histogram describing the number of intersections of sphere separators.
This is shown in \rfig{intersection-types}.
We compute the mean number of intersections for each area and normalize the result by dividing by the square root of the number of points.
We observe that the mean number of intersections is between $1.89 \sqrt n$ and $1.98 \sqrt n$ for the tiles tested.
Thus, in the examined areas, the number of intersections is on average $c \sqrt n$ for a small constant $c$.

\input{figures/intersection-growth}
\input{tables/intersection-growth-points-distribution}
Next, we examine how the number of intersections changes as the total area increases.
For this, we have conducted an experiment where we create 3 tiles of increasing sizes near the city of Herning.
When creating the tiles, we attempted to keep the number of ground and building points per $100 \kmsquared$ constant.
However, it is very difficult to keep them constant in practice due to the nature of the data.
Therefore, we expect that this will affect the number of intersections to some extend.
The distribution of point labels is described in \rtab{intersection-growth-point-distribution}.
The number of intersections for each tile is shown in \rfig{intersection-growth}.
As expected, we see a varying number of intersections between tiles.
However, we note that the normalized number of intersections remains small for all the tested tile sizes.

%% file: figures/intersection-types.tex
\begin{figure}[h]

  \begin{minipage}{0.33\linewidth}%
    \centering
    \includegraphics[width=1\linewidth]{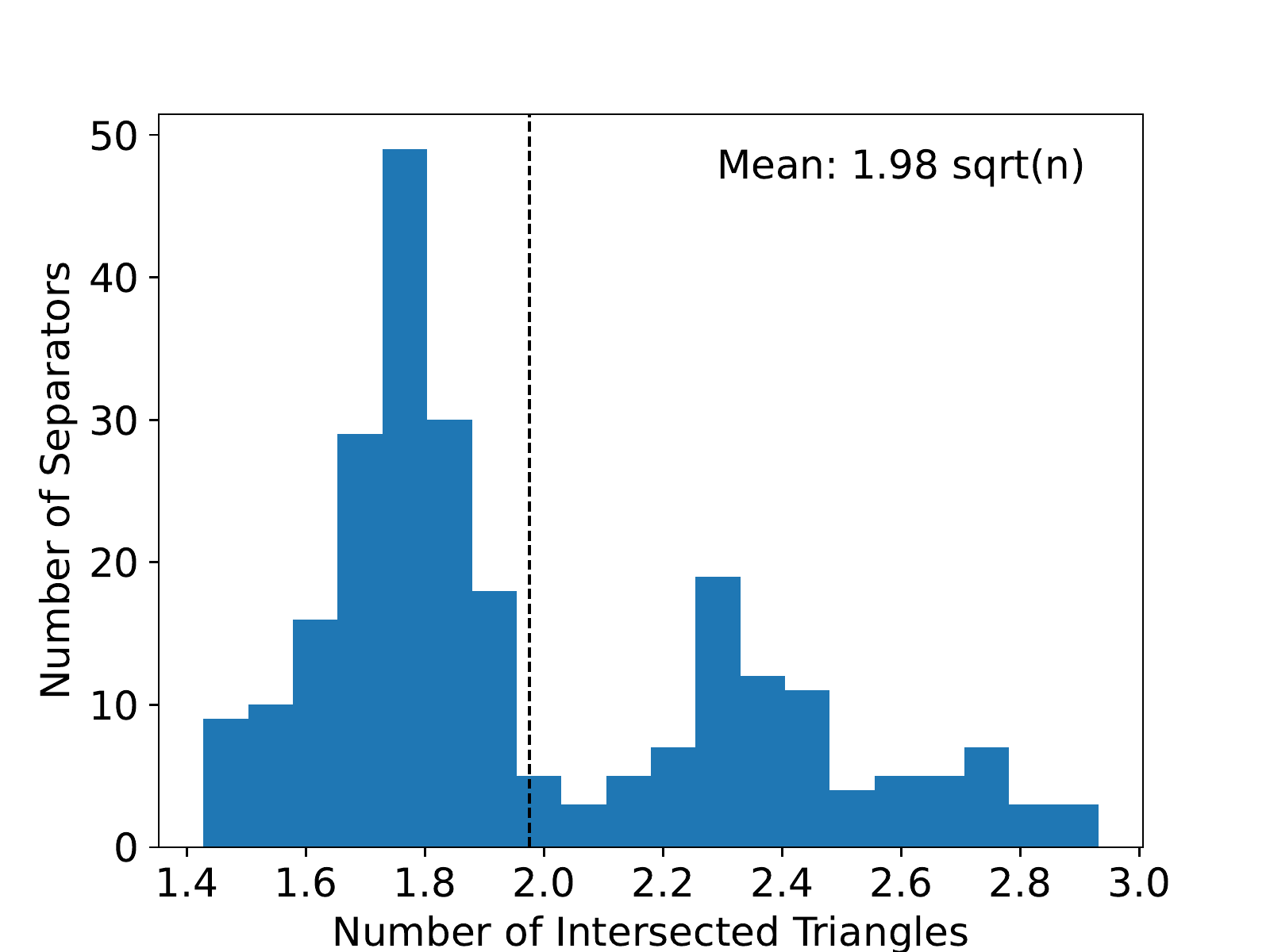}
    \subcaption{%
      Copenhagen (City)
      \label{fig:intersect-city}%
    }%
  \end{minipage}%
  \begin{minipage}{0.33\linewidth}%
    \centering
    \includegraphics[width=1\linewidth]{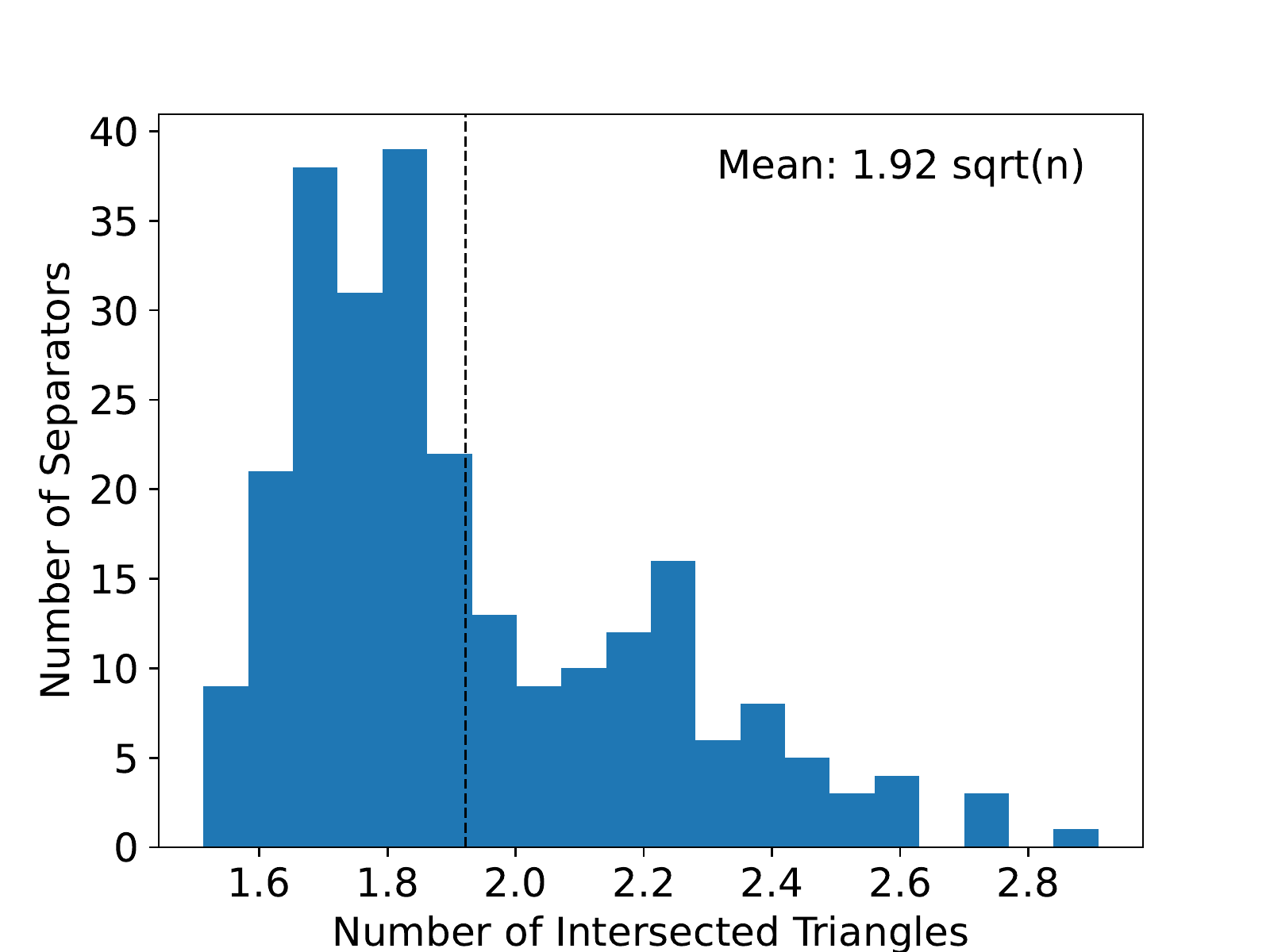}
    \subcaption{%
      Fjeldstervang (Fields)
      \label{fig:intersect-field}%
    }%
  \end{minipage}%
  \begin{minipage}{0.33\linewidth}%
    \centering
    \includegraphics[width=1\linewidth]{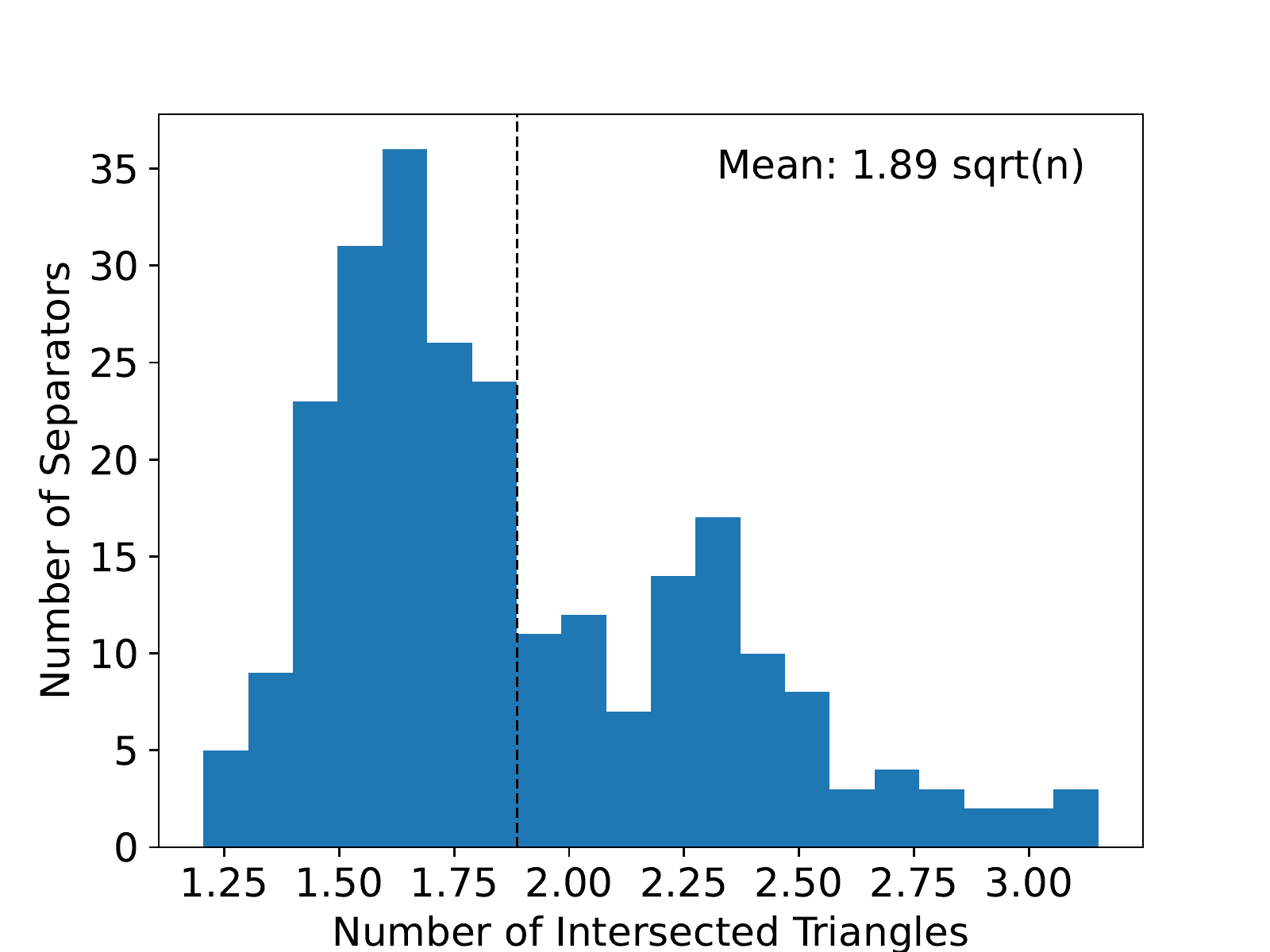}
    \subcaption{%
      Gammel Rye (Forest)
      \label{fig:intersect-forest}%
    }%
  \end{minipage}%
  \caption{%
    \normalfont
    Evaluating the number of intersected triangles on various types of terrain. Each area is $10 \km$ by $10 \km$
  }
  \label{fig:intersection-types}
\end{figure}

%% file: tables/intersection-types-points-distribution.tex
\begin{table}[h]
  \begin{center}
  \begin{tabular}{llrlrlr}
                      & \multicolumn{2}{c}{Copenhagen (city)} & \multicolumn{2}{c}{Fjeldstervang (fields)} & \multicolumn{2}{c}{Gammel Rye (forest)} \\ \hline
    Total points      & 1\,737\,352\,542 && 1\,016\,877\,728 && 1\,380\,507\,284 \\
    Ground            & 1\,003\,915\,809 &(57.78\%) & 726\,748\,491 &(71.47 \%) & 514\,504\,787 &(37.27 \%) \\
    Building          & 193\,699\,578 &(11.15\%) & 6\,462\,467 &(0.64 \%) & 8\,164\,875 &(0.59 \%) \\
    Low veg.    & 17\,072\,550 &(0.98\%) & 28\,417\,823 &(2.79 \%) & 23\,066\,640 &(1.67 \%) \\
    Medium veg. & 93\,116\,564 &(5.35\%) & 41\,787\,026 &(4.11 \%) & 66\,254\,550 &(4.80 \%) \\
    High veg.   & 379\,489\,187 &(21.84\%) & 210\,080\,703 &(20.66 \%) & 735\,216\,819 &(53.26 \%) \\
    Water             & 3\,619\,895 &(0.21\%) & 1\,746\,674 &(0.17 \%) & 29\,697\,070 &(2.15 \%) \\
  \end{tabular}
  \end{center}
  \caption{Description of point cloud labels for tiles of size $10 \km$ by $10 \km$. The most frequent point labels have been included.}
  \label{tab:intersection-types-point-distribution}
\end{table}

%% file: figures/intersection-growth.tex
\begin{figure}[h]

  \begin{minipage}{0.33\linewidth}%
    \centering
    \includegraphics[width=1\linewidth]{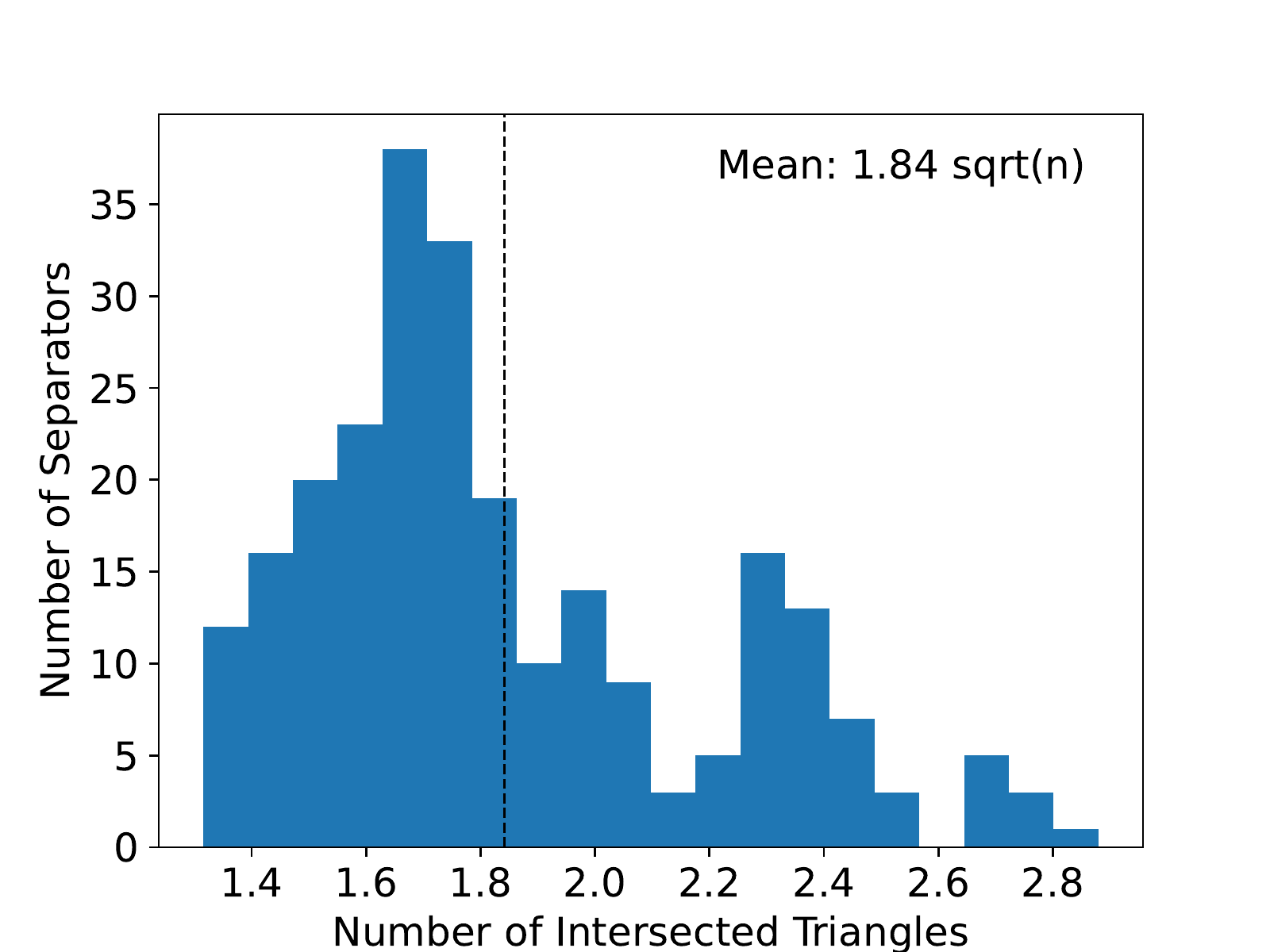}
    \subcaption{%
      A $100 \kmsquared$ tile.
      \label{fig:intersect-1x1}%
    }%
  \end{minipage}%
  \begin{minipage}{0.33\linewidth}%
    \centering
    \includegraphics[width=1\linewidth]{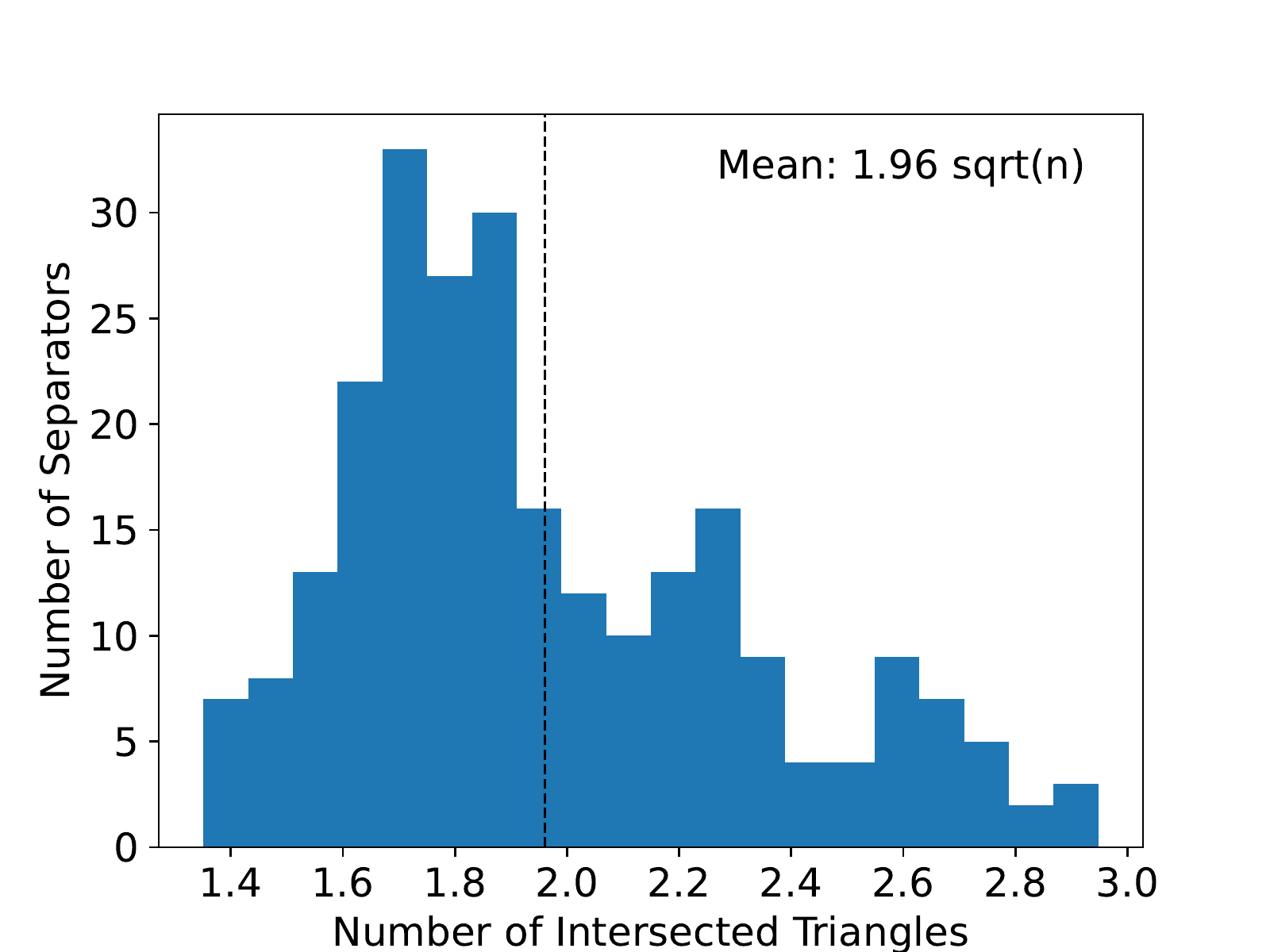}
    \subcaption{%
      A $400 \kmsquared$ tile.
      \label{fig:intersect-2x2}%
    }%
  \end{minipage}%
  \begin{minipage}{0.33\linewidth}%
    \centering
    \includegraphics[width=1\linewidth]{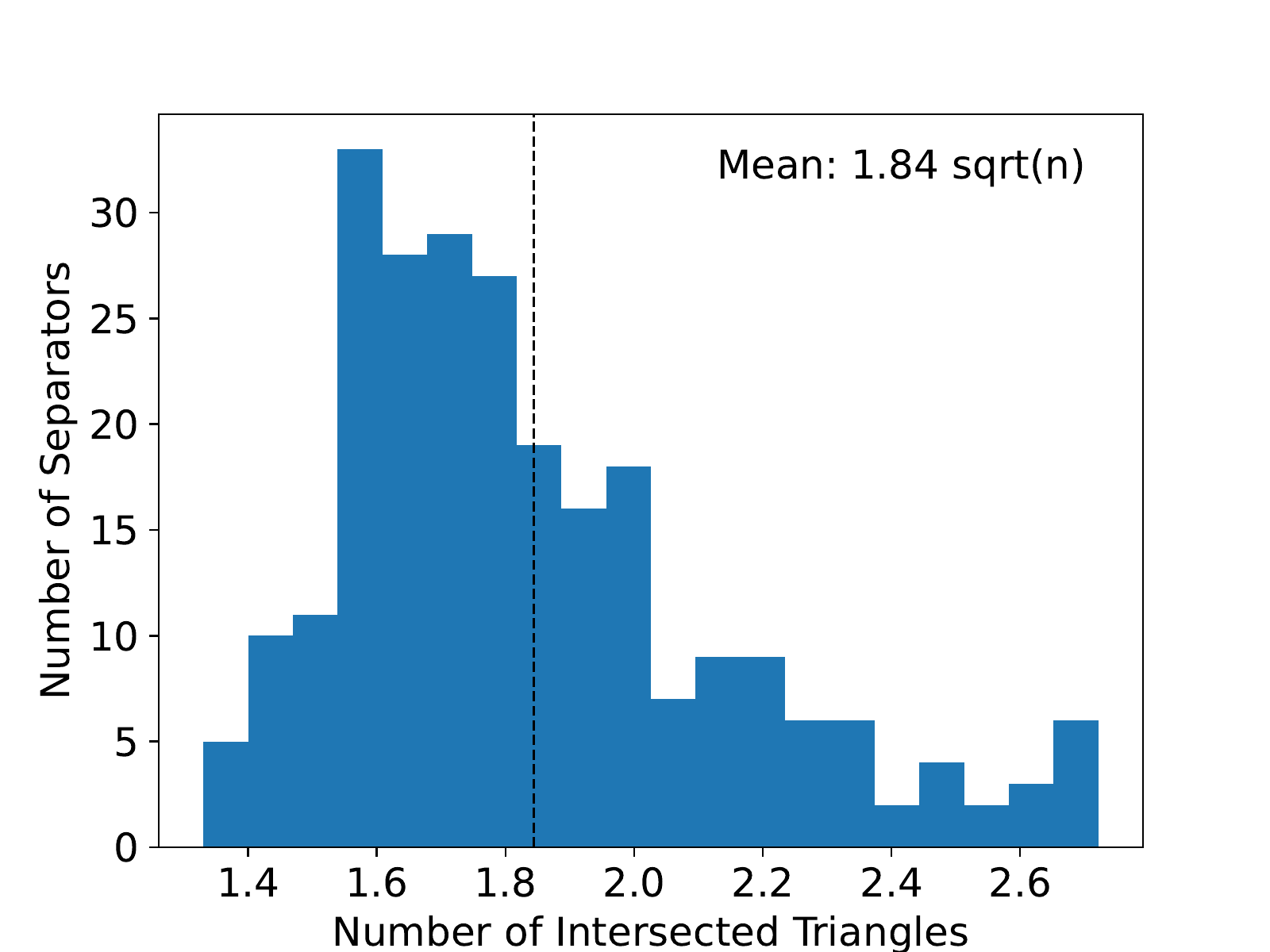}
    \subcaption{%
      A $900 \kmsquared$ tile.
      \label{fig:intersect-3x3}%
    }%
  \end{minipage}%
  \caption{%
    \normalfont
    Evaluating the number of intersected triangles when increasing total area.
  }
  \label{fig:intersection-growth}
\end{figure}

%% file: tables/intersection-growth-points-distribution.tex
\begin{table}[h]
  \begin{center}
  \begin{tabular}{llrlrlr}
                                      & \multicolumn{2}{c}{$100 \kmsquared$ Tile} & \multicolumn{2}{c}{$400 \kmsquared$ Tile} & \multicolumn{2}{c}{$900 \kmsquared$ Tile} \\ \hline
    Total points / $100 \kmsquared$ & \multicolumn{2}{l}{11\,841\,409}          & \multicolumn{2}{l}{10\,088\,500} & \multicolumn{2}{l}{10\,120\,894} \\
    Ground / $100 \kmsquared$       & 6\,733\,499   &(56.86\%)  & 6\,713\,922 &(66.55\%) & 6\,681\,457 &(66.02\%) \\
    Building / $100 \kmsquared$     & 114\,970    &(0.97\%)   & 202\,454  &(2.01\%) & 124\,327 &(1.23\%) \\
    Low veg. / $100 \kmsquared$     & 293\,570    &(2.48\%)   & 188\,550  &(1.87\%) & 204\,808 &(2.02\%) \\
    Medium veg. / $100 \kmsquared$ & 547\,759    &(4.63\%)   & 428\,343  &(4.25\%) & 418\,139 &(4.13\%) \\
    High veg. / $100 \kmsquared$    & 4\,103\,293   &(34.65\%)  & 2\,493\,932 &(24.72\%) & 2\,645\,132 &(26.14\%) \\
    Water / $100 \kmsquared$        & 19\,731     &(0.17\%)   & 30\,524   &(0.30\%) & 25\,666 &(0.25\%) \\
  \end{tabular}
  \end{center}
  \caption{Description of point cloud labels for increasing tile sizes. The most frequent point labels have been included.}
  \label{tab:intersection-growth-point-distribution}
\end{table}

%% file: appendix-large-sample.tex
\section{Algorithm with Larger Sample}
\label{sec:appendix-large-sample}
In this section, we modify the algorithm described in \rsec{algorithm} by increasing the size of the sample $\varsamplesystem$ and choosing sphere separators $h$ in a slightly different way.
By doing this, we can relax the assumption on $k$ to $k \leq |\varsystem|/\hat r$ and provide bounds on the number of boundary balls in $\varsystem$.

Define the set of classifiers $\varclassifiers_=$ as $\varclassifiers_= = \{ \vardisks(\varclassifier_=) \mid \varclassifier \in \varcircles \}$.
It follows from \rlem{vc-compose} and \rlem{vc-circle} that the range space $(\vardisks, \varclassifiers_=)$ have constant VC dimension.
Furthermore, we modify the definition of $\varclassifiers_l$ to include $\varclassifiers_=$ as follows:
\begin{align*}
  \varclassifiers_l &= \big\{ \varclassifier_1 \cap \cdots \cap \varclassifier_l \mid \varclassifier_1, \ldots, \varclassifier_l \in (\varclassifiers_\leq \cup \varclassifiers_\geq \cup \varclassifiers_=) \big\}\;.
\end{align*}
Using the same argument as the proof of \rlem{vc-compose-spheres}, we see that the VC-dimension of $(\vardisks, \varclassifiers_l)$ is $\BigO(l \log l)$.

Recall that the algorithm described in \rsec{algorithm:sample} forms a separator tree $\varseptree$ by recursively computing sphere separators $\varclassifier$ on $\varsamplesystem$.
The proof of the algorithm relied on the assumption that $\varsamplesystem$ is sampled such that $\varsamplesystem$ is a $(1/\hat r, \varepsilon)$-relative approximation of $\varsystem$, where $\varepsilon > 0$ is a constant.
This allowed us to argue about the size of the regions of $\varsystem$ within an error of $\BigO(|\varsystem|/\hat r)$.
However, we now wish to argue that the number of boundary balls in a region is $\BigO(\sqrt{k|\varsystem|/\hat r})$, which can be much smaller than $|\varsystem|/\hat r$.
Therefore, we let $p = \big({\sqrt{k|\varsystem|/\hat r}}\big)/{|\varsystem|} = \sqrt{k/(|\varsystem|\hat r)}$.
Assume that $\varsamplesystem$ is a $(p,\varepsilon)$-relative approximation of $\varsystem$ in the range space $(\vardisks, \varclassifiers_l)$.
Furthermore, we relax the assumption on $k$ to $k \leq |\varsystem|/\hat r$.

We modify the recursion of \rsec{algorithm:sample} as follows:
Recall that we divide $\varsamplesystem$ into regions of size at most $c \cdot |\varsamplesystem|/\hat r$.
Let $\varsamplesystem_i \subseteq \varsamplesystem$ and $\varsystem_i \subseteq \varsystem$ denote the balls in a node $i$ of the recursion.
We use \rthm{miller-separator} to obtain a sphere separator $h$ that with probability at least $1/2$ satisfies
\begin{align}
  \label{eqn:appendix-sample-inside} |\varsamplesystem_i(\varclassifier_<)|, |\varsamplesystem_i(\varclassifier_>)| &\leq \frac {10}{12} \cdot |\varsamplesystem_i|\;, \\
  \label{eqn:appendix-full-intersect} |\varsystem_i(\varclassifier_=)| &\leq c_1 \sqrt{k|\varsystem_i|}\;,
\end{align}
where $c_1>0$ is a constant.
Note that this uses expected $\Scan(|\varsamplesystem_i|)$ I/Os.
We use \reqn{appendix-full-intersect} and the assumption that $\varsamplesystem$ is a $(p,\varepsilon)$-relative approximation to obtain the following:
\begin{align*}
  |\varsamplesystem_i(\varclassifier_=)|
    &\leq (1+\varepsilon)\frac{|\varsamplesystem|}{|\varsystem|} \cdot c_1 \sqrt{k|\varsystem_i|} \\
    &\leq (1+\varepsilon)\frac{|\varsamplesystem|}{|\varsystem|} \cdot c_1 \sqrt{(1+\varepsilon)\frac{|\varsystem|}{|\varsamplesystem|} \cdot k|\varsamplesystem_i|} \\
    &= c_2 \sqrt{\frac{|\varsamplesystem|}{|\varsystem|} \cdot k |\varsamplesystem_i|}\;, \label{eqn:appendix-sample-intersect} \numberthis
\end{align*}
where $c_2 > 0$ is a constant.
Similar to before, we sample an expected constant number of separators $\varclassifier$, until \reqn{appendix-sample-inside} and \reqn{appendix-sample-intersect} are satisfied.
Note that the assumption $k \leq |\varsystem|/\hat r$ implies $k \frac{|\varsamplesystem|}{|\varsystem|} \leq \frac{|\varsamplesystem|}{\hat r}$.
Recalling that $|\varsamplesystem_i| \geq a = c\frac{|\varsamplesystem|}{\hat r}$, we rewrite as follows:
\begin{align*}
  |\varsamplesystem_i(\varclassifier_=)|
    &\leq c_2 \sqrt{\frac{|\varsamplesystem|}{\hat r} |\varsamplesystem_i|}
    \leq c_2 \sqrt{\frac{|\varsamplesystem_i|^2}{c}}
    \leq \frac{c_2}{\sqrt c} |\varsamplesystem_i| \;.
    \label{eqn:appendix-sample-bound} \numberthis
\end{align*}
Thus, for sufficiently large $c>0$, \reqn{appendix-sample-inside} and \reqn{appendix-sample-bound} implies $|\varsamplesystem_i(\varclassifier_\leq)| \leq \frac{11}{12} |\varsamplesystem_i|$ and $|\varsamplesystem_i(\varclassifier_\geq)| \leq \frac{11}{12} |\varsamplesystem_i|$.
In other words, the problem size becomes smaller by a constant factor and we can recursively compute separators $\varclassifier$ until $\varsamplesystem$ is divided into regions of size at most $c \cdot |\varsamplesystem|/\hat r$.

Recall that in \rsec{algorithm:sample}, we argued that $|\varsamplesystem_i(\varclassifier_=)| = \BigO\big(\sqrt{k |\varsamplesystem_i|}\big)$.
It follows that the above modification results in the number of intersections of $\varsamplesystem_i$ being smaller by a factor $\sqrt{|\varsamplesystem|/|\varsystem|}$.
By inserting this into the proofs described in \rsec{algorithm:total-boundary} and \rsec{algorithm:region-boundary}, we obtain that $\varsamplesystem$ is divided into $\BigO(\hat r)$ regions of size $\BigO(|\varsamplesystem|/\hat r)$ such that each region has $\BigO \left(\sqrt{\frac{|\varsamplesystem|}{|\varsystem|} \cdot k|\varsamplesystem|/\hat r}\right)$ boundary balls.
Furthermore, since $\varsamplesystem$ is a $(p,\varepsilon)$-relative approximation of $\varsystem$, we see that $\varsystem$ is divided into $\BigO(\hat r)$ regions of size $\BigO(|\varsystem|/\hat r)$ such that each region has $\BigO \big(\sqrt{k|\varsystem|/\hat r}\big)$ boundary balls.
It follows that the total number of boundary balls $\BigO \big(\sqrt{k|\varsystem|\hat r}\big)$.

In the description above, we assumed that that $\varsamplesystem$ is a $(p,\varepsilon)$-relative approximation of $\varsystem$.
However, if this is not the case, it might not be possible to obtain a classifier satisfying~\reqn{appendix-sample-intersect}.
To ensure that the algorithm terminates in this case, we let $m = \BigO(\log \hat r)$ and sample at most $m$ sphere separators in each node of the recursion.
We terminate if no separator satisfying \reqn{appendix-sample-inside} and \reqn{appendix-sample-intersect} is found.
Recall that in the case where $\varsamplesystem$ is $(p,\varepsilon)$-relative approximation, a sampled sphere separator $\varclassifier$ satisfies \reqn{appendix-sample-inside} and \reqn{appendix-sample-intersect} with probability at least $1/2$.
Thus, for a fixed node we fail to obtain a separator $h$ with probability $2^{-m}$ since the randomness in each sphere separator is independent.
Since the number of nodes in the separator tree is $\BigO(\hat r)$, at least one node in the tree will fail with probability at most $\BigO(\hat r) \cdot 2^{-m}$.
Thus,  all nodes succeed with at least constant probability by setting $m=c_4 \cdot \log \hat r$ for sufficiently large $c_4>0$.

We have now shown that, when given a $(p, \varepsilon)$-relative approximation of $\varsystem$, we can compute a separator tree $\varseptree$ that results in an $\hat r$-way separator for $\varsystem$ when applied to $\varsamplesystem$.
Given such a separator tree $\varseptree$, we can compute an $\hat r$-way separator of $\varsystem$ by recursively dividing the balls of $\varsystem$ using $\varseptree$ as described in \rsec{algorithm:complexity}.
Note that this uses $\BigO(\Scan(|\varsystem|))$ I/Os.
We now bound the expected total number of I/Os used when computing $\varseptree$.
Following the same line of the argumentation as in \rsec{algorithm}, the expected number of I/Os used in a level of the recursion is $\BigO(\Scan(|\varsamplesystem|) \log \hat r)$.
Note the additional log-factor that results from the case where $\varsamplesystem$ is not a $(p,\varepsilon)$-relative approximation of $\varsystem$.
Thus, the expected number of I/Os used in all $\BigO(\log \hat r)$ levels of the recursion is $\BigO(\Scan(|\varsamplesystem|) \log^2 \hat r)$.
It follows from \rlem{vc-sample-size}, that we obtain a $(p,\varepsilon)$-relative approximation $\varsamplesystem$ of $\varsystem$ with at least constant probability by sampling $\varsamplesystem$ of size
\begin{align*}
  |\varsamplesystem|
    = \BigO\left(\sqrt{\frac{|\varsystem|\hat r}{k}} \log \hat r \log \log \hat r \log\left(\sqrt{\frac{|\varsystem|\hat r}{k}}\right)\right)
    = \BigO\left(\sqrt{\frac{|\varsystem|\hat r}{k}} \log \hat r \log \log \hat r \log\frac{|\varsystem|}{k}\right) \;.
\end{align*}
Observe that this is significantly larger than the sample size described in \rsec{algorithm}.

Note that we can assume $\hat r \leq |\varsystem|/M$, since this will be sufficient to divide $\varsystem$ into regions of size $\BigO(M)$.
This follows from the observation that regions of size $\BigO(M)$ fit in memory after a constant number of applications of \rthm{miller-separator}.
Thus, they can be further divided using $\BigO(\Scan(|\varsystem|))$ I/Os in total.
We now obtain the following:
\begin{align*}
  \Scan(|\varsamplesystem|) \log^2 \hat r
  &= \BigO\left(\frac{\sqrt{|\varsystem|}}{B} \sqrt{\frac {\hat r}{k}} \log^3 \hat r \log \log \hat r \log\frac{|\varsystem|}{k}\right) \\
  &= \BigO\left(\frac{\sqrt{|\varsystem|}}{B} \sqrt{\frac{|\varsystem|}{Mk}} \log^3 \hat r \log \log \hat r \log\frac{|\varsystem|}{k}\right) \\
  &= \BigO\left(\frac{|\varsystem|}{B} \cdot \frac{\log^3 \hat r \log \log \hat r \log \frac{|\varsystem|}{k}}{\sqrt{Mk}} \right) \;.
\end{align*}
Recall that we wish to bound the expected total number of I/Os to $\BigO(\Scan(|\varsystem|))$ and that $\hat r \leq M/B$.
We now obtain the following lemma:
\begin{lemma}
  Given a $k$-ply neighborhood system $\varsystem$ in the plane and a parameter $\hat r \leq M/B$,
  an $\hat r$-way separator of $\varsystem$ can be computed using expected $\BigO \big(\Scan(|\varsystem|) \big)$ I/Os, assuming $k \leq |\varsystem|/\hat r$ and
  \begin{align*}
    \log^3 \hat r \log \log \hat r \log \frac{|\varsystem|}{k} = \BigO\Big(\sqrt{Mk}\Big) \;.
  \end{align*}
  \label{lem:appendix-base-case-large}
\end{lemma}
Given a parameter $r>0$, it follows that we can compute an $r$-way separator by recursively applying \rlem{appendix-base-case-large} for $\BigO(\log_{M/B}(r))$ levels.
This proves \rthm{large-sample}.